\documentclass[a4paper,12pt]{article}

\usepackage{color}
\usepackage{amssymb, amsmath}
\usepackage{hyperref}
\usepackage{tikz}

\newcommand\EFFACE[1]{}

\newtheorem{theorem}{Theorem}

\newtheorem{lemma}[theorem]{Lemma}
\newtheorem{corollary}[theorem]{Corollary}

\newtheorem{observation}[theorem]{Observation}

\newenvironment{proof}{
\par
\noindent {\bf Proof.}\rm}{\mbox{}\hfill$\square$\par\vskip 3mm}

\def\NNN{\mathbb{N}}

\def\diam{{\rm diam}}
\def\rad{{\rm rad}}
\def\cost{{\rm cost}}

\makeatletter
\let\@fnsymbol\@arabic
\makeatother

\begin{document}


\newcommand\CTzero[2]{
\node[scale=0.7,draw,circle,fill=black] (v0) at (#1,#2){};
}

\newcommand\CTun[2]{
\node[scale=0.7,draw,circle,fill=black] (v0) at (#1,#2){};
\node[scale=0.7,draw,circle,fill=black] (v1) at (#1,#2-2){};
\draw[thick] (v0) to (v1);
}

\newcommand\CTzeroplus[2]{
\node[scale=0.7,draw,circle,fill=black] (v0) at (#1,#2){};
\node[scale=0.7,draw,circle,fill=black] (v1) at (#1,#2-2){};
\draw[thick,dotted] (v0) to (v1);
}

\newcommand\CTunplus[2]{
\node[scale=0.7,draw,circle,fill=black] (v0) at (#1,#2){};
\node[scale=0.7,draw,circle,fill=black] (v1) at (#1-0.2,#2-2){};
\draw[thick] (v0) to (v1);
\draw[thick,dotted] (#1+0.1,#2-2) -- ++(0.4,0);
}

\newcommand\CTdeux[2]{
\node[scale=0.7,draw,circle,fill=black] (v0) at (#1,#2){};
\node[scale=0.7,draw,circle,fill=black] (v1) at (#1-0.3,#2-2){};
\node[scale=0.7,draw,circle,fill=black] (v2) at (#1+0.3,#2-2){};
\draw[thick] (v0) to (v1);
\draw[thick] (v0) to (v2);
}

\newcommand\CTdeuxmoins[2]{
\node[scale=0.7,draw,circle,fill=black] (v0) at (#1,#2){};
\node[scale=0.7,draw,circle,fill=black] (v1) at (#1-0.3,#2-2){};
\node[scale=0.7,draw,circle,fill=black] (v2) at (#1+0.3,#2-2){};
\draw[thick] (v0) to (v1);
\draw[thick,dashed] (v0) to (v2);
}

\newcommand\CTdeuxplus[2]{
\node[scale=0.7,draw,circle,fill=black] (v0) at (#1,#2){};
\node[scale=0.7,draw,circle,fill=black] (v1) at (#1-0.5,#2-2){};
\node[scale=0.7,draw,circle,fill=black] (v3) at (#1+0.5,#2-2){};
\draw[thick] (v0) to (v1);
\draw[thick] (v0) to (v3);
\draw[thick,dotted] (#1-0.15,#2-2) -- ++(0.4,0);
}

\newcommand\CTtrois[2]{
\node[scale=0.7,draw,circle,fill=black] (v0) at (#1,#2){};
\node[scale=0.7,draw,circle,fill=black] (v1) at (#1-0.6,#2-2){};
\node[scale=0.7,draw,circle,fill=black] (v2) at (#1,#2-2){};
\node[scale=0.7,draw,circle,fill=black] (v3) at (#1+0.6,#2-2){};
\draw[thick] (v0) to (v1);
\draw[thick] (v0) to (v2);
\draw[thick] (v0) to (v3);
}

\newcommand\CTtroisplus[2]{
\node[scale=0.7,draw,circle,fill=black] (v0) at (#1,#2){};
\node[scale=0.7,draw,circle,fill=black] (v1) at (#1-0.8,#2-2){};
\node[scale=0.7,draw,circle,fill=black] (v2) at (#1-0.2,#2-2){};
\node[scale=0.7,draw,circle,fill=black] (v3) at (#1+0.8,#2-2){};
\draw[thick] (v0) to (v1);
\draw[thick] (v0) to (v2);
\draw[thick] (v0) to (v3);
\draw[thick,dotted] (#1+0.15,#2-2) -- ++(0.4,0);
}

\newcommand\CTzeroL[3]{
\node[scale=0.7,draw,circle,fill=black] (v0) at (#1,#2){};
\node[below] at (#1,#2+1) {#3};
}

\newcommand\CTunL[3]{
\node[scale=0.7,draw,circle,fill=black] (v0) at (#1,#2){};
\node[scale=0.7,draw,circle,fill=black] (v1) at (#1,#2-2){};
\draw[thick] (v0) to (v1);
\node[above] at (#1,#2-3) {#3};
}

\newcommand\CTunplusL[3]{
\node[scale=0.7,draw,circle,fill=black] (v0) at (#1,#2){};
\node[scale=0.7,draw,circle,fill=black] (v1) at (#1-0.2,#2-2){};
\node[above] at (#1-0.2,#2-3) {#3};
\draw[thick] (v0) to (v1);
\draw[thick,dotted] (#1+0.1,#2-2) -- ++(0.4,0);
}

\newcommand\CTdeuxL[4]{
\node[scale=0.7,draw,circle,fill=black] (v0) at (#1,#2){};
\node[scale=0.7,draw,circle,fill=black] (v1) at (#1-0.3,#2-2){};
\node[scale=0.7,draw,circle,fill=black] (v2) at (#1+0.3,#2-2){};
\draw[thick] (v0) to (v1);
\draw[thick] (v0) to (v2);
\node[above] at (#1-0.3,#2-3) {#3};
\node[above] at (#1+0.3,#2-3) {#4};
}

\newcommand\CTdeuxplusL[4]{
\node[scale=0.7,draw,circle,fill=black] (v0) at (#1,#2){};
\node[scale=0.7,draw,circle,fill=black] (v1) at (#1-0.5,#2-2){};
\node[scale=0.7,draw,circle,fill=black] (v3) at (#1+0.5,#2-2){};
\draw[thick] (v0) to (v1);
\draw[thick] (v0) to (v3);
\node[above] at (#1-0.5,#2-3) {#3};
\node[above] at (#1+0.5,#2-3) {#4};
\draw[thick,dotted] (#1-0.15,#2-2) -- ++(0.4,0);
}

\newcommand\CTtroisL[5]{
\node[scale=0.7,draw,circle,fill=black] (v0) at (#1,#2){};
\node[scale=0.7,draw,circle,fill=black] (v1) at (#1-0.6,#2-2){};
\node[scale=0.7,draw,circle,fill=black] (v2) at (#1,#2-2){};
\node[scale=0.7,draw,circle,fill=black] (v3) at (#1+0.6,#2-2){};
\draw[thick] (v0) to (v1);
\draw[thick] (v0) to (v2);
\draw[thick] (v0) to (v3);
\node[above] at (#1-0.6,#2-3) {#3};
\node[above] at (#1,#2-3) {#4};
\node[above] at (#1+0.6,#2-3) {#5};
}

\newcommand\CTtroisplusL[5]{
\node[scale=0.7,draw,circle,fill=black] (v0) at (#1,#2){};
\node[scale=0.7,draw,circle,fill=black] (v1) at (#1-0.8,#2-2){};
\node[scale=0.7,draw,circle,fill=black] (v2) at (#1-0.2,#2-2){};
\node[scale=0.7,draw,circle,fill=black] (v3) at (#1+0.8,#2-2){};
\draw[thick] (v0) to (v1);
\draw[thick] (v0) to (v2);
\draw[thick] (v0) to (v3);
\node[above] at (#1-0.8,#2-3) {#3};
\node[above] at (#1-0.2,#2-3) {#4};
\node[above] at (#1+0.8,#2-3) {#5};
\draw[thick,dotted] (#1+0.15,#2-2) -- ++(0.4,0);
}

\def\CATERPILLAR{
\CTun{0}{0}
\CTzero{2}{0}
\CTdeux{4}{0}
\CTun{6}{0}
\CTun{8}{0}
\CTdeux{10}{0}
\CTun{12}{0}
\CTzero{14}{0}
\CTtrois{16}{0}
\draw[thick] (0,0) to (16,0);
%
}

\def\CATERPILLARSAVE{
\node[draw,circle,fill=black] (v0) at (0,0){};
\node[draw,circle,fill=black] (v1) at (2,0){};
\node[draw,circle,fill=black] (v2) at (4,0){};
\node[draw,circle,fill=black] (v3) at (6,0){};
\node[draw,circle,fill=black] (v4) at (8,0){};
\node[draw,circle,fill=black] (v5) at (10,0){};
\node[draw,circle,fill=black] (v6) at (12,0){};
\node[draw,circle,fill=black] (v7) at (14,0){};
\node[draw,circle,fill=black] (v8) at (16,0){};

\node[draw,circle,fill=black] (f01) at (0,-2){};
\node[draw,circle,fill=black] (f21) at (3.4,-2){};
\node[draw,circle,fill=black] (f22) at (4.6,-2){};
\node[draw,circle,fill=black] (f31) at (6,-2){};
\node[draw,circle,fill=black] (f41) at (8,-2){};
\node[draw,circle,fill=black] (f51) at (9.4,-2){};
\node[draw,circle,fill=black] (f52) at (10.6,-2){};
\node[draw,circle,fill=black] (f61) at (12,-2){};
\node[draw,circle,fill=black] (f81) at (15.2,-2){};
\node[draw,circle,fill=black] (f82) at (16,-2){};
\node[draw,circle,fill=black] (f83) at (16.8,-2){};

\draw[thick] (v0) to (v8);

\draw[thick] (v0) to (f01);
\draw[thick] (v2) to (f21);
\draw[thick] (v2) to (f22);
\draw[thick] (v3) to (f31);
\draw[thick] (v4) to (f41);
\draw[thick] (v5) to (f51);
\draw[thick] (v5) to (f52);
\draw[thick] (v6) to (f61);
\draw[thick] (v8) to (f81);
\draw[thick] (v8) to (f82);
\draw[thick] (v8) to (f83);
} 


\title{On the Broadcast Independence Number of Caterpillars}

\author{Messaouda AHMANE~\thanks{Faculty of Mathematics, Laboratory L'IFORCE, University of Sciences and Technology
Houari Boumediene (USTHB), B.P.~32 El-Alia, Bab-Ezzouar, 16111 Algiers, Algeria.}
\and Isma BOUCHEMAKH~\footnotemark[1]
\and \'Eric SOPENA~\thanks{Univ. Bordeaux, Bordeaux INP, CNRS, LaBRI, UMR 5800, F-33400 Talence, France.}
}

\maketitle

\abstract{
Let $G$ be a simple undirected graph.
A broadcast on $G$ is
a function $f : V(G)\rightarrow\NNN$ such that $f(v)\le e_G(v)$ holds for every vertex $v$ of $G$, 
where $e_G(v)$ denotes the eccentricity of $v$ in $G$, that is, the maximum distance from $v$ to any other vertex of $G$.
The cost of $f$ is the value $\cost(f)=\sum_{v\in V(G)}f(v)$.
A broadcast $f$ on $G$ is independent if for every two distinct vertices $u$ and $v$ in $G$, $d_G(u,v)>\max\{f(u),f(v)\}$,
where $d_G(u,v)$ denotes the distance between $u$ and $v$ in $G$.
The broadcast independence number of $G$ is then defined as the maximum cost of an independent broadcast on $G$.
 
In this paper, we study independent broadcasts of caterpillars and give an explicit formula for the
 broadcast independence number of caterpillars having no pair of adjacent trunks, 
a trunk being an internal spine vertex with degree~2.
}

\medskip



\section{Introduction}

All the graphs we consider in this paper are simple and loopless undirected graphs. 
We denote by $V(G)$ and $E(G)$ the set of vertices and the set of edges of a graph $G$, respectively.

For any two vertices $u$ and $v$ of $G$,  the \emph{distance} $d_G(u,v)$ between $u$ and $v$ in $G$
is the length (number of edges) of a shortest path joining $u$ and $v$.
The \emph{eccentricity} $e_G(v)$ of a vertex $v$ in $G$
the maximum distance from $v$ to any other vertex of $G$. 
The minimum eccentricity in $G$ is the \emph{radius} $\rad(G)$ of $G$, while the maximum eccentricity in $G$ is the
\emph{diameter} $\diam(G)$ of $G$. 
Two vertices $u$ and $v$ with $d_G(u, v) = \diam(G)$ are said to be \emph{antipodal}.

A function $f : V(G)\rightarrow\{0,\dots,\diam(G)\}$ is a \emph{broadcast} if for every vertex $v$ of $G$, $f(v)\le e_G(v)$.
The value $f(v)$ is  called the \emph{$f$-value of $v$}.
Given a broadcast $f$ on $G$, an \emph{$f$-broadcast vertex} is a vertex $v$ with $f(v)>0$.
The set of all $f$-broadcast vertices is denoted $V_f^+$.
If $u\in V_f^+$ is a broadcast vertex, $v\in V(G)$ and $d_G(u,v)\le f(u)$, we say that \emph{$u$ $f$-dominates $v$}.
In particular, every $f$-broadcast vertex $f$-dominates itself.
The \emph{cost} $\cost(f)$ of a broadcast $f$ on $G$ is given by
$$\cost(f)=\sum_{v\in V(G)}f(v)=\sum_{v\in V_f^+}f(v).$$

A broadcast $f$ on $G$ is a \emph{dominating broadcast} if every vertex of $G$
is $f$-dominated by some vertex of $V_f^+$.
The minimum cost of a dominating broadcast on $G$ is the \emph{broadcast domination number} of $G$,
denoted $\gamma_b(G)$.
A 
broadcast $f$ on $G$ is an \emph{independent broadcast} if every $f$-broadcast vertex 
is $f$-dominated only by itself.
The maximum cost of an independent broadcast on $G$ is the \emph{broadcast independence number} of $G$,
denoted $\beta_b(G)$.
An independent broadcast on $G$ with cost $\beta$ is an independent \emph{$\beta$-broadcast}.
An independent $\beta_b(G)$-broadcast on $G$ is an \emph{optimal} independent broadcast.
Note here that any optimal independent broadcast is necessarily a dominating broadcast.

The notions of broadcast domination and broadcast independence were introduced by D.J.~Erwin
in his Ph.D. thesis~\cite{E01} under the name of \emph{cost domination} and \emph{cost independence}, respectively.
During the last decade,
broadcast domination has been investigated by several authors, 
see e.g.~\cite{BHHM04,BB10,BS11,BS09,CHM11,DDH09,HL06,HM09,LM15,MW13,S08,SK14},
while independent broadcast domination has attracted much less attention.

In particular, Seager considered in~\cite{S08} broadcast domination of caterpillars.
She characterized caterpillars with broadcast domination number equal to their domination number,
and caterpillars with broadcast domination number equal to their radius.
Blair, Heggernes, Horton and Manne proposed in~\cite{BHHM04} an $O(nr)$-algorithm for computing
the broadcast domination number of a tree of order $n$ with radius $r$.

However, determining the independent broadcast number of trees seems to be a difficult problem.
We propose in this paper a first step in this direction, by studying a subclass of the
class of caterpillars. Recall that a caterpillar is a tree such that deleting all its pendent
vertices leaves a simple path called the spine. The subclass we will consider is the subclass of caterpillars
having no pair of adjacent trunks, 
a trunk being an internal spine vertex with degree~2.

We now review a few results on independent broadcast numbers.
Let $G$ be a graph  and $A\subset V(G)$, $|A|\ge 2$, be a set of pairwise antipodal vertices in $G$.
The function $f$ defined by $f(u)=\diam(G)-1$ for every vertex $u\in A$,
and $f(v)=0$ for every vertex $v\not\in A$,
is clearly an independent $|A|(\diam(G)-1)$-broadcast on $G$. 

\begin{observation}[Dunbar, Erwin, Haynes, Hedetniemi and Hedetniemi~\cite{DEHHH06}]\mbox{}\\
For every graph $G$ of order at least 2
and every set $A\subset V(G)$, $|A|\ge 2$, of pairwise antipodal vertices in $G$, $\beta_b(G)\ge |A|(\diam(G)-1)$.
In particular, for every tree $T$, $\beta_b(T)\ge 2(\diam(G)-1)$.
\label{obs:2(d-1)}
\end{observation}

An independent broadcast $f$ on a graph $G$ is \emph{maximal independent} if there
is no independent broadcast $f'\neq f$ such that $f'(v)\ge f(v)$ for every
vertex $v\in V(G)$.
In~\cite{E01}, D.J.~Erwin proved the following result (see also~\cite{DEHHH06}).

\begin{theorem}[Erwin~\cite{E01}]\mbox{}\\
Let $f$ be an independent broadcast on $G$.
If $V^+_f=\{v\}$, then $f$ is maximal independent if and only if $f(v)=e_G(v)$.
If $|V^+_f|\ge 2$, then $f$ is maximal independent if and only if the following two
conditions are satisfied:
\begin{enumerate}
\item $f$ is dominating, and
\item for every $v\in V^+_f$, $f(v) = \min\big\{d_G(v,u):\ u\in V^+_f\setminus\{v\}\big\} - 1$.
\end{enumerate}
\label{th:maximality}
\end{theorem}

Erwin proved that $\beta_b(P_n)=2(n-2)=2(\diam(P_n)-1)$ for every path $P_n$
of length $n\ge 3$~\cite{E01}.
In~\cite{BZ14}, Bouchemakh and Zemir determined the independent broadcast number of 
square grids.

\begin{theorem}[Bouchemakh and Zemir~\cite{BZ14}]\mbox{}\\
Let $G_{m,n}$ denote the square grid with $m$ rows and $n$ columns, $m\ge 2$, $n\ge 2$. We then have:
\begin{enumerate}
\item $\beta_b(G_{m,n})=2(m+n-3)=2(\diam(G_{m,n})-1)$ if $m\le 4$,
\item $\beta_b(G_{5,5})=15$, $\beta_b(G_{5,6})=16$, and
\item $\beta_b(G_{m,n})=\left\lceil\frac{mn}{2}\right\rceil$ for every $m,n$, $5\le m\le n$, $(m,n)\neq (5,5),(5,6)$.
\end{enumerate}
\end{theorem}

In this paper, we determine the broadcast independence
number of  caterpillars
having no pair of adjacent trunks.
The paper is organised as follows. We introduce in the next section the
main definitions and a few preliminary results on independent broadcasts
of caterpillars.
We then consider in Section~\ref{sec:no-adjacent-trunks} the case of caterpillars
having no pair of adjacent trunks
and prove our main result, which gives an explicit
formula for 
the broadcast independence number of such caterpillars.
We finally propose a few directions for future research in Section~\ref{sec:discussion}.

\section{Preliminaries}
\label{sec:preliminaries}

Let $G$ be a graph and $H$ be a subgraph of $G$.
Since $d_H(u,v)\ge d_G(u,v)$ for every two vertices $u,v\in V(H)$,
every independent broadcast $f$ on $G$ satisfying $f(u)\le e_H(u)$ for every
vertex $u\in V(H)$ is an independent broadcast on $H$.
Hence we have:

\begin{observation}
If $H$ is a subgraph of $G$ and $f$ is an independent broadcast on $G$ 
satisfying $f(u)\le e_H(u)$ for every vertex $u\in V(H)$, then the restriction
$f_H$ of $f$ to $V(H)$ is an independent broadcast on $H$.
\label{obs:subgraph}
\end{observation}

\medskip

A {\em caterpillar of length} $k\ge 0$ is a tree 
such that removing all leaves gives a path of length $k$, called the \emph{spine}.
Following the terminology of~\cite{S08}, a non-leaf vertex is called a \emph{spine vertex}
and, more precisely, a \emph{stem} 
if it is adjacent to a leaf and a \emph{trunk} otherwise.
A leaf adjacent to a stem $v$ is a \emph{pendent neighbour} of $v$.
We will always draw caterpillars with the spine on a horizontal line,
so that we can speak about the leftmost of rightmost spine vertex of a caterpillar.

Note that a caterpillar of length~0 is nothing but a star $K_{1,n}$, for some $n\ge 1$.
The independent broadcast number of a star is easy to determine.

\begin{observation}
For every integer $n\ge 1$, $\beta_b(K_{1,n})=n$.
\label{obs:star}
\end{observation}

Indeed, an optimal broadcast $f$ of $K_{1,n}$ is obtained by setting to~1 the $f$-value
of every pendent vertex of $K_{1,n}$, if $n>1$, or of one of the two vertices of $K_{1,1}$.
Therefore, in the rest of the paper, we will only consider caterpillars of length $k\ge 1$.

Let $\NNN^*=\NNN\setminus\{0\}$.
We denote by $CT(\lambda_0,\dots,\lambda_k)$, $k\ge 1$, 
with $(\lambda_0,\dots,\lambda_k)\in\NNN^*\times\NNN^{k-1}\times\NNN^*$,  
the caterpillar of length $k\ge 1$ with spine $v_0\dots v_k$ such that each spine vertex
$v_i$ has $\lambda_i$ pendent neighbours.
Note that for any caterpillar $CT$ of length $k\ge 1$, $\diam(CT)=k+2$.
For every $i$ such that $\lambda_i>0$, $0\le i\le k$, we denote by
$\ell_i^1,\dots,\ell_i^{\lambda_i}$ the pendent neighbours of~$v_i$. 
Moreover, we denote by $CT[a,b]$, $0\le a\le b\le k$, the subgraph 
of $CT$ induced by vertices $v_a,\dots,v_b$ and their pendent neighbours.
The caterpillar $CT(1,0,2,1,1,2,1,0,3)$ is depicted in Figure~\ref{fig:caterpillar}.

\begin{figure}
\begin{center}
\begin{tikzpicture}[scale=0.9,domain=0:17,x=0.7cm,y=0.7cm]
\CATERPILLAR
\end{tikzpicture}
\caption{\label{fig:caterpillar}The caterpillar $CT(1,0,2,1,1,2,1,0,3)$}
\end{center}
\end{figure}
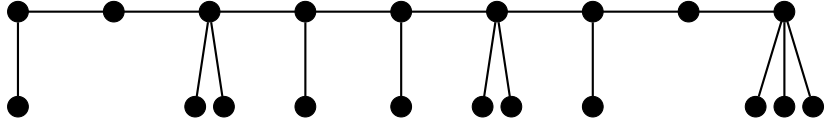

Let $f$ be an independent broadcast on a caterpillar $CT=CT(\lambda_0,\dots,\lambda_k)$.
We denote by $f^*$ the associated mapping from $\{v_0,\dots,v_k\}$ to $\NNN$ 
defined by 
$$f^*(v_i)=f(v_i)+\sum_{j=1}^{j=\lambda_i}f(\ell_i^j),\ \mbox{if}\  \lambda_i>0,\ \ 
\mbox{and}\ \ f^*(v_i)=f(v_i)\ \mbox{otherwise,}$$
for every $i$, $0\le i\le k$. 
Intuitively speaking, when $\lambda_i>0$, $f^*(v_i)$ gives the ``weight''
of the star-graph consisting of the vertex $v_i$ together with its pendent neighbours.

We will say that two independent broadcasts $f_1$ and $f_2$ on $CT$ are \emph{similar}
whenever $f_1^*=f_2^*$. Observe that any two similar independent broadcasts have the same cost.

From Observation~\ref{obs:2(d-1)}, we get that $\beta_b(CT)\ge 2(k+1)$
for every caterpillar $CT=CT(\lambda_0,\dots,\lambda_k)$.
In particular, the function $f_c$ on $V(CT)$ defined by
$f_c(\ell_0^1)=f_c(\ell_k^1)=k+1$ and $f_c(u)=0$ for every vertex
$u\in V(CT)\setminus\{\ell_0^1,\ell_k^1\}$ is an independent
broadcast on $CT$ with cost $2(k+1)$.

In the following, we will call any independent broadcast $f$ similar to $f_c$ 
and such that $|V^+_f|=2$ a
\emph{canonical} independent broadcast.

The following lemma shows that, for any caterpillar $CT=CT(\lambda_0,\dots,\lambda_k)$,
no independent 
broadcast  $f$ on $CT$ with $f(v)>0$ for some stem $v$ 
can be optimal.

\begin{lemma}
If $CT=CT(\lambda_0,\dots,\lambda_k)$ is a caterpillar of length $k\ge 1$ and $f$ is an
independent broadcast on $CT$ with $f(v_i)>0$ for some stem $v_i$, $0\le i\le k$, then
there exists an independent broadcast $f'$ on $CT$ with $\cost(f')>\cost(f)$.
\label{lem:stem-0}
\end{lemma}

\begin{proof}
Since $f(v_i)>0$ and $f$ is an independent broadcast, we have $f(\ell_i^j)=0$ for every $j$,
$1\le j\le\lambda_i$.
Consider the function $f'$ defined by 
$f'(v_i)=0$, $f'(\ell_i^1)=f(v_i)+1$ and $f'(u)=f(u)$ for every vertex $u\in V(CT)\setminus\{v_i,\ell_i^1\}$.
Since $d_{CT}(\ell_i^1,u)=d_{CT}(v_i,u)+1$ for every vertex $u\in V(CT)\setminus\{\ell_i^1\}$, we get that
$f'$ is an independent broadcast on $CT$. Moreover, we clearly have $\cost(f')=\cost(f)+1$.
\end{proof}

The following lemma shows that for every optimal independent broadcast on a
caterpillar, at least one pendent vertex of each of the end-vertices of the spine
is a broadcast vertex.


\begin{lemma}
Let $CT=CT(\lambda_0,\dots,\lambda_k)$ be a caterpillar of length $k\ge 1$.
If $f$ is an optimal independent broadcast on $CT$, then 
$f^*(v_0)-f(v_0)\neq 0$ and $f^*(v_k)-f(v_k)\neq 0$.
\label{lem:end-leaves}
\end{lemma}

\begin{proof}
%
We know by Lemma~\ref{lem:stem-0} that $f(v_0)=0$.
Suppose, contrary to the statement of the lemma, that $f(\ell_0^j)=0$ for every $j$, $1\le j\le\lambda_0$.
Let $u$ be the $f$-broadcast vertex that dominates $\ell_0^1$ and let $f(u)=x$.
By Lemma~\ref{lem:stem-0}, $u$ is either a leaf or a trunk.

If $u$ is a leaf, say $u=\ell_i^j$, $1\le i\le k$, $1\le j\le\lambda_i$,
let $f'$ be the mapping defined by $f'(\ell_0^1)=x+i$, $f'(u)=0$ and $f'(u')=f(u')$
for every vertex $u'\in V(CT)\setminus\{\ell_0^1,u\}$.
Note that every vertex which was $f$-dominated by $u$ is now $f'$-dominated by $\ell_0^1$.
The mapping $f'$ is thus an independent $(\cost(f)+i)$-broadcast on $CT$, contradicting the optimality of~$f$.

If $u$ is a trunk, say $u=v_i$, $1\le i\le k-1$, we similarly define a mapping $f'$
by letting $f'(\ell_0^1)=x+i+1$, $f'(u)=0$ and $f'(u')=f(u')$
for every vertex $u'\in V(CT)\setminus\{\ell_0^1,u\}$.
The mapping $f'$ is thus an independent $(\cost(f)+i+1)$-broadcast on $CT$, again contradicting the optimality of~$f$.

The case $f(\ell_k^j)=0$ for every $j$, $1\le j\le\lambda_k$, follows by symmetry.
\end{proof}

Observe that Lemma~\ref{lem:end-leaves} can be extended to trees as follows:

\begin{lemma}
Let $T$ be tree and $T'$ be a subtree of $T$, of order at least 2, with root $r$.
Let $f$ be an optimal independent broadcast on $T$.
If $r$ is an $f$-broadcast vertex, then $T'$ contains at least one other $f$-broadcast vertex.
In particular, if $T'$ is a subtree of height~1 (that is, $e_{T'}(r)=1$), then $f(r)=0$.
\label{lem:end-leaves-tree}
\end{lemma}

\begin{proof}
Suppose to the contrary that $f(r)>0$ and $f(u)=0$ for every vertex $u\in V(T')\setminus\{r\}$.
Let $t'=e_{T'}(r)$ and $\overline{t'}=e_{T-(T'-r)}(r)$. 

If $f(r)<t'$, the independent broadcast $f'$ given by $f'(v)=f(r)$ for some vertex $v$ in $T'$ with $d_{T'}(r,v)=t'$
and $f'(u)=f(u)$ for every vertex $u\in V'(T)\setminus\{v\}$ is such that $\cost(f')=\cost(f)+f(r)$, 
contradicting the optimality of $f$.

If $f(r)\ge\overline{t'}$, then $r$ is the unique $f$-broadcast vertex, which implies
$\cost(f)<2(\diam(T)-1)$, again contradicting the optimality of $f$ by Observation~\ref{obs:2(d-1)}.

Hence $\overline{t'}>f(r)\ge t'$.
Let now $v$ be any neighbour of $r$ in $T'$.
Since $\overline{t'}>f(r)\ge t'$, we have $e_T(v)=e_T(r)+1=\overline{t'}+1>f(r)+1$.
The function $f'$ defined by $f'(r)=0$, $f'(v)=f(r)+1$
and $f'(u)=f(u)$ for every vertex $u\in V(T)\setminus\{r,v\}$ is therefore
an independent broadcast on $T$ with $\cost(f')=\cost(f)+1$, contradicting the optimality of $f$.

This completes the proof.
\end{proof}

\section{Caterpillars with no pair of adjacent trunks}
\label{sec:no-adjacent-trunks}

In this section we determine the broadcast independence number
of caterpillars with no pair of adjacent trunks.
We first introduce some notation and useful lemmas.

We say that an independent broadcast $f$ of a caterpillar $CT$ is
an \emph{optimal non-canonical} independent broadcast on $CT$ if 
\begin{itemize}
\item[(i)] $|V^+_f|\neq 2$ or $f^*\neq f_c^*$ ($f$ is non-canonical), and
\item[(ii)] for every independent broadcast $f'$ on $CT$ with $|V^+_{f'}|\neq 2$ or $f'^*\neq f_c^*$,
$\cost(f)\ge\cost(f')$ ($f$ is optimal among all non-canonical independent broadcasts).
\end{itemize}


Let $CT=CT(\lambda_0,\dots,\lambda_k)$ be a caterpillar of length $k\ge 1$ with no pair of adjacent trunks.
We denote by 
$$\lambda(CT)=\sum_{i=0}^{i=k}\lambda_i$$
the number of leaves of $CT$, and by 
$$\tau(CT)=|\{i\ |\ 1\le i\le k-1\ \mbox{and}\ \lambda_i=0\}|$$
the number of trunks of $CT$.

We will compute the broadcast independence number
of a caterpillar with no pair of adjacent trunks by counting the number of some specific \emph{patterns}. 
More precisely,
we say that a pattern of length $p+1$, $\Pi=\pi_0\dots\pi_p$, $p\ge 0$, $\pi_i\in\NNN$ for every $i$, $0\le i\le p$, 
\emph{occurs} in a caterpillar $CT=CT(\lambda_0,\dots,\lambda_k)$
if there exists an index $i_0$,
$0\le i_0\le k-p$, such that $CT[i_0,i_0+p]=CT(\pi_0,\dots,\pi_p)$, that is,
$\lambda_{i_0+j}=\pi_j$ for every $j$, $0\le j\le p$.
We will also say that the caterpillar $CT$ \emph{contains} the pattern~$\Pi$
and that the subgraph $CT(\lambda_{i_0},\dots,\lambda_{i_0+p})$ of $CT$ is an \emph{occurrence}
of the pattern~$\Pi$.
For instance, the caterpillar $CT(1,0,2,1,1,2,1,0,3)$, depicted on Figure~\ref{fig:caterpillar}, 
contains once the pattern 211 and twice the pattern 10.

We now extend the notation for patterns as follows:
\begin{itemize}
\item By $\pi_i^+$, we mean a spine vertex having at least $\pi_i$ pendent neighbours;
\item By $\pi_i^-$, we mean a spine vertex having at most $\pi_i$ pendent neighbours;
\item By $[\Pi$, we mean that the pattern $\Pi$ occurs and starts at the leftmost stem $v_0$,
\item By $\Pi]$, we mean that the pattern $\Pi$ occurs and ends at the rightmost stem $v_k$,
\item By $\{\Pi,\Pi'\}$, we mean either the pattern $\Pi$ or the pattern $\Pi'$.
\item By $\pi_0(\pi_1\pi_2)^{+r}\pi_3$, we mean a \emph{maximal} pattern of the form 
$$\pi_0\pi_1\pi_2\pi_3\ \ \mbox{or}\ 
\ \pi_0\underbrace{\pi_1\pi_2\dots\pi_1\pi_2}_{\mbox{{\footnotesize $r$ times, $r\ge 2$}}}\pi_3,$$
where maximal here means that the subpattern $\pi_1\pi_2$ is repeated at least once and as many times as possible.
\item By $\pi_0(\pi_1\pi_2)^{*r}\pi_3$, we mean a \emph{maximal} pattern of the form 
$$\pi_0\pi_3,\ \pi_0\pi_1\pi_2\pi_3\ \ \mbox{or}\ 
\ \pi_0\underbrace{\pi_1\pi_2\dots\pi_1\pi_2}_{\mbox{{\footnotesize $r$ times, $r\ge 2$}}}\pi_3,$$
where maximal here means that the subpattern $\pi_1\pi_2$ is repeated 
as many times as possible.
\end{itemize}
We can also combine these notations, so that, for instance, $\pi_i^+]$ denotes that the rightmost stem
 $v_k$ has at least $\pi_i$ pendent neighbours,
and $\{\pi_i,[\}\Pi$ denotes either the pattern $\pi_i\Pi$ or the pattern~$[\Pi$.

One can check that the caterpillar $CT(1,0,2,1,1,2,1,0,3)$, depicted on Figure~\ref{fig:caterpillar},
contains once each of the four patterns
$[1$, $3]$, $2^+]$ and $2111^+$, twice the pattern $0\{2,3\}$,
and thrice the pattern $1^+1^+1^+$.
On the other hand, the caterpillar $CT(1,0,2,0,2,0,2,1,0,3)$ contains only once the pattern $1^+0(20)^{+r}1^+$,
namely on the sub-caterpillar $CT(1,0,2,0,2,0,2)$ with explicit pattern $1020202$.

For any pattern $\Pi$ and any caterpillar $CT$, we will denote by $\#_{CT}(\Pi)$
the number of occurrences of the pattern $\Pi$ in $CT$.
Moreover, if $M$ is an occurrence of $\Pi$ in $CT$, we define the value
$$\alpha_1(M)=\max\{0,\#_{M}(1)-1\},$$
that is, the number of stems $v_i$ in $M$ with $\lambda_i=1$ minus 1---or 0 if $M$ contains no such stem---,
and the value
$$\alpha_2(M)=\alpha_1(M)+\#_{M}([1^+)+\#_{M}(1^+]),$$
that is, $\alpha_1(M)$
plus 0, 1 or 2, depending on whether $M$ contains no end-vertex of $CT$, one end-vertex of $CT$ or both end-vertices
of $CT$, respectively.

We then extend the functions $\alpha_1$ and $\alpha_2$ to the whole caterpillar $CT$ by setting
$$\alpha_1(CT;\Pi)=\sum_{\mbox{\footnotesize $M$ occurrence of $\Pi$}}\alpha_1(M)$$
and
$$\alpha_2(CT;\Pi)=\sum_{\mbox{\footnotesize $M$ occurrence of $\Pi$}}\alpha_2(M).$$

\begin{figure}
\begin{center}
\begin{tikzpicture}[scale=0.9,x=0.7cm,y=0.7cm]
\CTun{0}{0}
\CTunplus{2}{0}
\draw[thick] (0,0) to (2,0);
\draw[thick,dashed] (2,0) to (3,0);
\node at (1.5,-3) {$[11^+$};
\end{tikzpicture}
\hskip 2cm
\begin{tikzpicture}[scale=0.9,x=0.7cm,y=0.7cm]
\CTunplus{0}{0}
\CTun{2}{0}
\CTunplus{4}{0}
\draw[thick] (0,0) to (4,0);
\draw[thick,dashed] (-1,0) to (0,0);
\draw[thick,dashed] (4,0) to (5,0);
\node at (2,-3) {$1^+11^+$};
\end{tikzpicture}
\mbox{}\\
\vskip 0.5cm
\begin{tikzpicture}[scale=0.9,x=0.7cm,y=0.7cm]
\CTunplus{0}{0}
\CTdeuxmoins{2}{0}
\CTzero{4}{0}
\CTdeuxmoins{6}{0}
\CTzero{10}{0}
\CTdeuxmoins{12}{0}
\CTunplus{14}{0}
\draw[thick,dashed] (-1,0) to (0,0);
\draw[thick] (0,0) to (6,0);
\draw[thick,dotted] (6,0) to (10,0);
\draw[thick] (10,0) to (14,0);
\draw[thick,dashed] (14,0) to (15,0);
\node at (7,-3) {$1^+2^-(02^-)^{+r}1^+$};
\end{tikzpicture}
\mbox{}\\
\vskip 0.5cm
\begin{tikzpicture}[scale=0.9,x=0.7cm,y=0.7cm]
\CTzero{0}{0}
\CTdeuxmoins{2}{0}
\CTzero{4}{0}
\CTdeuxmoins{6}{0}
\CTzero{8}{0}
\CTdeuxmoins{10}{0}
\CTzero{14}{0}
\CTdeuxmoins{16}{0}
\draw[thick] (-1,0) to (10,0);
\draw[thick,dotted] (10,0) to (14,0);
\draw[thick] (14,0) to (16,0);
\node at (8,-3) {$02^-(02^-)^{*r}]$};
\end{tikzpicture}

\caption{\label{fig:patterns}Sample patterns involved in the definition of $\beta^*(CT)$}
\end{center}
\end{figure}
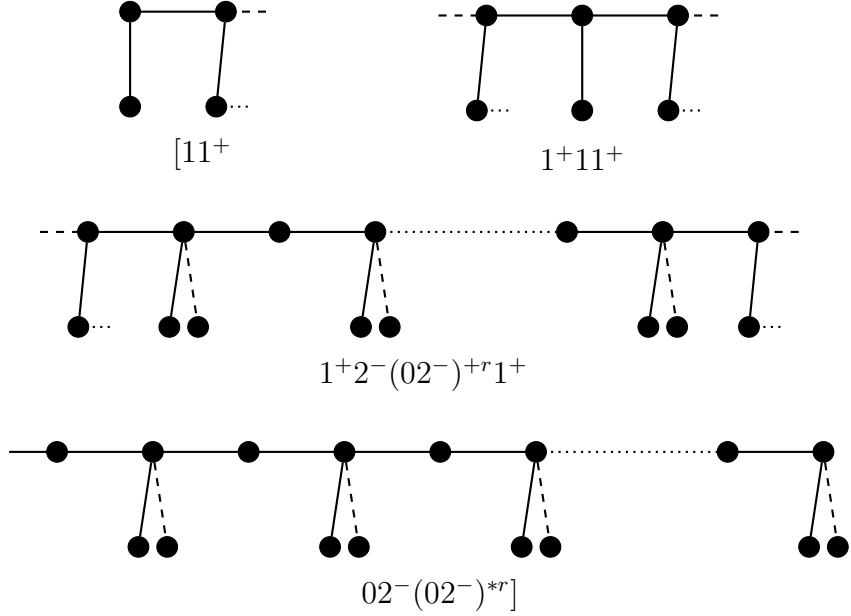

Finally, for any caterpillar $CT$, we define the value $\beta^*(CT)$ as follows:
$$\begin{array}{rcl}
\beta^*(CT) & = & \lambda(CT)\ +\ \tau(CT)\ +\ \#_{CT}(\{1^+,[\}1\{1^+,]\})\ +\ \alpha_1(CT;1^+2^-(02^-)^{+r}1^+)\\ 
& & +\ \alpha_2(CT;02^-(02^-)^{*r}0)\ +\ \alpha_2(CT;[2^-(02^-)^{*r}0)\ +\ \alpha_2(CT;02^-(02^-)^{*r}]).
\end{array}
$$

Sample patterns involved in the above formula
 are illustrated in Figure~\ref{fig:patterns}.
 In the figure,
a pattern with a line to the left or right hand side of its spine cannot occur
at the left or right end of the caterpillar, respectively.
A pattern with a dashed line to the left or right hand side of its spine can occur 
at the left or right end of the caterpillar, respectively, or in the middle of the caterpillar.
A dashed edge is an optional edge (used for pattern $2^-$, corresponding to a spine vertex with either
one or two pendent neighbours).

Let us say that two \emph{distinct} occurrences of patterns \emph{overlap} if they share a common vertex.
Due to the specific structure of the patterns used in the above formula
(and, in particular, of the maximality of the number of repetitions of subpatterns
of the form $\Pi^{+r}$ or $\Pi^{*r}$), 
we have the following:

\begin{observation}
In every caterpillar $CT$ of length $k\ge 1$,  
\begin{enumerate}
\item 
no occurrence of the pattern $02^-(02^-)^{*r}0$
can overlap with an occurrence of a pattern
$\{1^+,[\}1\{1^+,]\}$, $1^+2^-(02^-)^{+r}1^+$, $02^-(02^-)^{*r}0$,
$[2^-(02^-)^{*r}0$ or $02^-(02^-)^{*r}]$,
\item 
no occurrence of the pattern $[2^-(02^-)^{*r}0$
can overlap with an occurrence of a pattern
$\{1^+,[\}1\{1^+,]\}$, or $1^+2^-(02^-)^{+r}1^+$,
\item 
no occurrence of the pattern $02^-(02^-)^{*r}]$
can overlap with an occurrence of a pattern
$\{1^+,[\}1\{1^+,]\}$ or $1^+2^-(02^-)^{+r}1^+$, 
\item if two occurrences of the patterns $[2^-(02^-)^{*r}0$ and $02^-(02^-)^{*r}]$ overlap,
then $CT$ is a caterpillar with pattern $[2^-(02^-)^{*r}]$. 
%
\end{enumerate}
\label{obs:overlap}
\end{observation}

\begin{figure}
\begin{center}

\begin{tikzpicture}[scale=0.9,x=0.7cm,y=0.7cm]
\CTunL{0}{0}{1}
\CTunplusL{2}{0}{1}
\draw[thick] (0,0) to (2,0);
\draw[thick,dashed] (2,0) to (3,0);
\node at (5,-1) {$\longrightarrow$};
\CTunL{7}{0}{2}
\CTunplusL{9}{0}{1}
\draw[thick] (7,0) to (9,0);
\draw[thick,dashed] (9,0) to (10,0);
\end{tikzpicture}
\mbox{}\\
\vskip 0.5cm
\begin{tikzpicture}[scale=0.9,x=0.7cm,y=0.7cm]
\CTunplusL{0}{0}{1}
\CTunL{2}{0}{1}
\CTunplusL{4}{0}{1}
\draw[thick] (0,0) to (4,0);
\draw[thick,dashed] (-1,0) to (0,0);
\draw[thick,dashed] (4,0) to (5,0);
\node at (7,-1) {$\longrightarrow$};
\CTunplusL{9}{0}{1}
\CTunL{11}{0}{2}
\CTunplusL{13}{0}{1}
\draw[thick] (9,0) to (13,0);
\draw[thick,dashed] (8,0) to (9,0);
\draw[thick,dashed] (13,0) to (14,0);
\node at (7,-4) {(a) From $f_1$ to $f_2$};
\end{tikzpicture}
\mbox{}\\
\vskip 0.5cm
\begin{tikzpicture}[scale=0.9,x=0.7cm,y=0.7cm]
\CTunplus{0}{0}
\CTdeuxL{2}{0}{1}{1}
\CTzeroL{4}{0}{1}
\CTunL{6}{0}{1}
\CTzeroL{8}{0}{1}
\CTdeuxL{10}{0}{1}{1}
\CTunplus{12}{0}
\draw[thick] (0,0) to (12,0);
\end{tikzpicture}
\mbox{}\\
\vskip 0.5cm
\begin{tikzpicture}[scale=0.9,x=0.7cm,y=0.7cm]
\node at (-2,-1) {$\longrightarrow$};
\CTunplus{0}{0}
\CTdeuxL{2}{0}{1}{1}
\CTzeroL{4}{0}{0}
\CTunL{6}{0}{3}
\CTzeroL{8}{0}{0}
\CTdeuxL{10}{0}{1}{1}
\CTunplus{12}{0}
\draw[thick] (0,0) to (12,0);
\node at (5,-4) {(b) From $f_2$ to $f_3$, pattern $1^+201021^+$, $\cost(f'_3)=\cost(f'_2)+(1-1)$};
\end{tikzpicture}
\mbox{}\\
\vskip 0.5cm
\begin{tikzpicture}[scale=0.9,x=0.7cm,y=0.7cm]
\CTunplus{0}{0}
\CTdeuxL{2}{0}{1}{1}
\CTzeroL{4}{0}{1}
\CTunL{6}{0}{1}
\CTzeroL{8}{0}{1}
\CTdeuxL{10}{0}{1}{1}
\CTzeroL{12}{0}{1}
\CTunL{14}{0}{1}
\CTzeroL{16}{0}{1}
\CTunL{18}{0}{1}
\CTunplus{20}{0}
\draw[thick] (0,0) to (20,0);
\end{tikzpicture}
\mbox{}\\
\vskip 0.5cm
\begin{tikzpicture}[scale=0.9,x=0.7cm,y=0.7cm]
\node at (-2,-1) {$\longrightarrow$};
\CTunplus{0}{0}
\CTdeuxL{2}{0}{1}{1}
\CTzeroL{4}{0}{0}
\CTunL{6}{0}{3}
\CTzeroL{8}{0}{0}
\CTdeuxL{10}{0}{3}{0}
\CTzeroL{12}{0}{0}
\CTunL{14}{0}{3}
\CTzeroL{16}{0}{0}
\CTunL{18}{0}{2}
\CTunplus{20}{0}
\draw[thick] (0,0) to (20,0);
\node at (9,-4) {(c) From $f_2$ to $f_3$, pattern $1^+2010201011^+$, $\cost(f'_3)=\cost(f'_2)+(3-1)$};
\end{tikzpicture}

\caption{\label{fig:f1tof4a}Proof of Lemma~\ref{lem:beta-star}: from $f_1$ to $f_3$}
\end{center}
\end{figure}

We first prove that every caterpillar with no pair of adjacent trunks
admits an independent broadcast $f$ with $\cost(f)=\beta^*(CT)$.

\begin{lemma}
Every caterpillar $CT=CT(\lambda_0,\dots,\lambda_k)$ of length $k\ge 1$, with no pair of adjacent trunks,
admits an independent broadcast $f$ with $\cost(f)=\beta^*(CT)$.
\label{lem:beta-star}
\end{lemma}

\begin{proof}
We will construct a sequence of independent broadcasts $f_1$, $\dots$, $f_4$,
step by step, such that $\cost(f_4)=\beta^*(CT)$.
Each independent broadcast $f_i$, $2\le i\le 4$, is obtained by possibly modifying 
the independent broadcast $f_{i-1}$ and is such that $\cost(f_i)\ge\cost(f_{i-1})$.
Moreover, for each independent broadcast $f_i$, $1\le i\le 4$, we will have $f_i(v)=0$
whenever $v$ is a stem.
These modifications are illustrated in Figures~\ref{fig:f1tof4a} and~\ref{fig:f1tof4b}, using
the same drawing conventions as in Figure~\ref{fig:patterns}.
Only useful broadcast values are given in these figures.
These figures should help the reader to see that all the proposed modifications
lead to a new valid independent broadcast.

\medskip

\noindent{\em Step 1.}
Let $f_1$ be the mapping defined by
$f_1(v)=1$ if  $v$ is a pendent vertex or a trunk, and $f_1(v)=0$ otherwise.
Clearly, $f_1$ is an independent broadcast on $CT$ with 
$$\cost(f_1)=\lambda(CT)+\tau(CT).$$

\noindent{\em Step 2.}
Let $f_2$ be the mapping defined by
$f_2(v)=2$ if $v=\ell_i^1$ for some $i$, $0\le i\le k$, such that (i) $\lambda_i=1$,
(ii) $i=0$ or $\lambda_{i-1}\ge 1$, and
(iii) $i=k$ or $\lambda_{i+1}\ge 1$,
and $f_2(v)=f_1(v)$ otherwise (see Figure~\ref{fig:f1tof4a}(a)).
Again, $f_2$ is an independent broadcast on $CT$ with 
$$\cost(f_2)=\cost(f_1)+\#_{CT}(\{1^+,[\}1\{1^+,]\}).$$

\noindent{\em Step 3.}
Suppose that $CT$ contains the pattern $1^+2^-(02^-)^{+r}1^+$, of length $2r+3$,
and let $M=CT[i_0,i_0+2r+2]$ be the corresponding occurrence of this pattern.
We thus have $f_2(v)=1$ for every trunk of $M$ and for every pendent neighbour
of a stem vertex $v_j$ on $M$ with $i_0+1\le j\le i_0+2r+1$.
Hence, the cost of the restriction $f'_2$ of $f_2$ to $M$ is
$$\cost(f'_2)=f_2^*(v_{i_0})+\lambda(M[i_0+1,i_0+2r+1])+\tau(M)+f_2^*(v_{i_0+2r+2}).$$
We modify $f_2$ as follows, to obtain $f_3$.
If the subgraph $M[i_0+1,i_0+2r+1]$ contains a stem vertex $v_i$ with $\lambda_i=1$,
we let 
\begin{itemize}
\item $f_3(\ell_{i_0+1}^1)=2$ if $\lambda_{i_0+1}=1$,
\item $f_3(\ell_{i_0+2r+1}^1)=2$ if $\lambda_{i_0+2r+1}=1$,
\item $f_3(\ell_{i_0+2j+1}^1)=3$ (and $f_3(\ell_{i_0+2j+1}^2)=0$ if $\lambda_{i_0+2j+1}=2$)
for every $j$, $1\le j\le r-1$,
\item  $f_3(v_{i_0+2j})=0$ for every $j$, $1\le j\le r$,
\end{itemize}
 (see Figure~\ref{fig:f1tof4a}(b) and~(c)).
The cost of the restriction $f'_3$ of $f_3$ on $M$ is then
$$\cost(f'_3)=\cost(f'_2)+\max\{0,\#_{M[i_0+1,i_0+2r+1]}(1)-1\}=\cost(f'_2)+\alpha_1(M).$$
By Observation~\ref{obs:overlap}, two occurrences of the pattern $1^+2^-(02^-)^{+r}1^+$ can only overlap on their end-vertices.
Therefore, doing the above modification for every occurrence of the pattern $1^+2^-(02^-)^{+r}1^+$ in $M$, 
the so-obtained independent broadcast $f_3$ satisfies
$$\cost(f_3)=\cost(f_2)+\alpha_1(CT).$$

\noindent{\em Step 4.}
Suppose first that $CT$ contains the pattern $02^-(02^-)^{*r}0$, of length $2r+3$,
and let $M=CT[i_0,i_0+2r+2]$, $i_0\ge 1$, $i_0+2r+2\le k-1$, be the corresponding occurrence of this pattern.
We thus have $f_2(v)=1$ for every trunk of $M$ and for every pendent neighbour
of a stem vertex $v_j$ on $M$ with $i_0+1\le j\le i_0+2r+1$.
Hence, the cost of the restriction $f'_3$ of $f_3$ to $M$ is
$$\cost(f'_3)=f_3^*(v_{i_0})+\lambda(M)+\tau(M[i_0+1,i_0+2r+1])+f_3^*(v_{i_0+2r+2}).$$
We modify $f_3$ as follows, to obtain $f_4$.
If the subgraph $M[i_0+1,i_0+2r+1]$ contains a stem vertex $v_i$ with $\lambda_i=1$,
we let 
\begin{itemize}
\item $f_4(\ell_{i_0+2j+1}^1)=3$ (and $f_4(\ell_{i_0+2j+1}^2)=0$ if $\lambda_{i_0+2j+1}=2$)
for every $j$, $0\le j\le r$,
\item  $f_4(v_{i_0+2j})=0$ for every $j$, $0\le j\le r$,
\end{itemize}
 (see Figure~\ref{fig:f1tof4b}(a)).
The cost of the restriction $f'_4$ of $f_4$ on $M$ is then
$$\cost(f'_4)=\cost(f'_3)+\max\{0,\#_{M}(1)-1\}=\cost(f'_3)+\alpha_2(M).$$

Suppose now that $CT$ contains the pattern $[2^-(02^-)^{*r}0$, %
of length $2r+2$,
and let $M=CT[0,2r+1]$ be the corresponding occurrence of this pattern.
Doing the same type of modification as above (see Figure~\ref{fig:f1tof4b}(b)),
the cost of the restriction $f'_4$ of $f_4$ on $M$ is then
$$\cost(f'_4)=\cost(f'_3)+\max\{0,\#_{M}(1)-1\}+1=\cost(f'_3)+\alpha_2(M).$$

Finally, if $CT$ contains the pattern $02^-(02^-)^{*r}]$ and $CT$ is not
a caterpillar with pattern $[2^-(02^-)^{*r}]$, the same type of modification
leads to the same property.

By Observation~\ref{obs:overlap}, no two occurrences of the patterns $02^-(02^-)^{*r}0$ and
$[2^-(02^-)^{*r}0$ (or $02^-(02^-)^{*r}0$ and $02^-(02^-)^{*r}]$)
can  overlap.
Therefore, doing the above modification for every occurrence of these patterns in $M$, 
the so-obtained independent broadcast $f_4$ satisfies
$$\cost(f_4)=\cost(f_3)+\alpha_2(CT)=\beta^*(CT).$$


This completes the proof.
\end{proof}

\begin{figure}
\begin{center}

\begin{tikzpicture}[scale=0.9,x=0.7cm,y=0.7cm]
\CTzeroL{0}{0}{1}
\CTdeuxL{2}{0}{1}{1}
\CTzeroL{4}{0}{1}
\CTunL{6}{0}{1}
\CTzeroL{8}{0}{1}
\CTdeuxL{10}{0}{1}{1}
\CTzeroL{12}{0}{1}
\CTunL{14}{0}{1}
\CTzeroL{16}{0}{1}
\CTunL{18}{0}{1}
\CTzeroL{20}{0}{1}
\draw[thick] (-1,0) to (21,0);
\end{tikzpicture}
\mbox{}\\
\vskip 0.5cm
\begin{tikzpicture}[scale=0.9,x=0.7cm,y=0.7cm]
\node at (-2,-1) {$\longrightarrow$};
\CTzeroL{0}{0}{0}
\CTdeuxL{2}{0}{3}{0}
\CTzeroL{4}{0}{0}
\CTunL{6}{0}{3}
\CTzeroL{8}{0}{0}
\CTdeuxL{10}{0}{3}{0}
\CTzeroL{12}{0}{0}
\CTunL{14}{0}{3}
\CTzeroL{16}{0}{0}
\CTunL{18}{0}{3}
\CTzeroL{20}{0}{0}
\draw[thick] (-1,0) to (21,0);
\node at (9,-4) {(a) From $f_3$ to $f_4$, pattern $02010201010$, $\cost(f'_4)=\cost(f'_3)+(3-1)+0$};
\end{tikzpicture}

\mbox{}\\
\vskip 0.5cm
\begin{tikzpicture}[scale=0.9,x=0.7cm,y=0.7cm]
\CTdeuxL{2}{0}{1}{1}
\CTzeroL{4}{0}{1}
\CTunL{6}{0}{1}
\CTzeroL{8}{0}{1}
\CTdeuxL{10}{0}{1}{1}
\CTzeroL{12}{0}{1}
\CTunL{14}{0}{1}
\CTzeroL{16}{0}{1}
\CTdeuxL{18}{0}{1}{1}
\CTzeroL{20}{0}{1}
\draw[thick] (2,0) to (21,0);
\end{tikzpicture}
\mbox{}\\
\vskip 0.5cm
\begin{tikzpicture}[scale=0.9,x=0.7cm,y=0.7cm]
\node at (0,-1) {$\longrightarrow$};
\CTdeuxL{2}{0}{3}{0}
\CTzeroL{4}{0}{0}
\CTunL{6}{0}{3}
\CTzeroL{8}{0}{0}
\CTdeuxL{10}{0}{3}{0}
\CTzeroL{12}{0}{0}
\CTunL{14}{0}{3}
\CTzeroL{16}{0}{0}
\CTdeuxL{18}{0}{3}{0}
\CTzeroL{20}{0}{0}
\draw[thick] (2,0) to (21,0);
\node at (10.5,-4) {(b) From $f_3$ to $f_4$, pattern $[2010201020$, $\cost(f'_4)=\cost(f'_3)+(2-1)+1$};
\end{tikzpicture}

%

\caption{\label{fig:f1tof4b}Proof of Lemma~\ref{lem:beta-star}: from $f_3$ to $f_4$}
\end{center}
\end{figure}

The next lemma shows that if $f$ is an optimal non-canonical independent broadcast
on a caterpillar $CT$ with no pair of adjacent trunks, with $\cost(f)>2(\diam(CT)-1)$,
then there exists an optimal non-canonical independent broadcast $\tilde{f}$ on $CT$
such that the $\tilde{f}$-values of the pendent neighbours of $v_0$ and $v_k$
only depend on the values of $\lambda_0,\lambda_1$ and $\lambda_{k-1},\lambda_k$, respectively:

\begin{lemma}
Let $CT=CT(\lambda_0,\dots,\lambda_k)$ be a caterpillar of length $k\ge 1$,
with no pair of adjacent trunks.
If $f$ is an optimal non-canonical independent broadcast on $CT$ 
with $\cost(f)>2(\diam(CT)-1)$, then there exists an optimal non-canonical independent broadcast $\tilde{f}$ on $CT$, 
thus with $\cost(\tilde{f})=\cost(f)$,
such that, for every $i\in\{0,k\}$, we have
\begin{enumerate}
\item if $\lambda_i=1$ and $\lambda_{i'}\ge 1$, then $\tilde{f}(\ell_i^1)= 2$,
\item if $\lambda_i=1$ and $\lambda_{i'}=0$,
then $\tilde{f}(\ell_i^1)= 3$,
\item if $\lambda_i=2$ and $\lambda_{i'}\ge 1$, then $\tilde{f}(\ell_i^1)=\tilde{f}(\ell_i^2)= 1$,
\item if $\lambda_i=2$ and $\lambda_{i'}=0$,
then $\tilde{f}(\ell_i^1)= 3$ and $\tilde{f}(\ell_i^2)=0$,
\item if $\lambda_i\ge 3$, then $\tilde{f}(\ell_i^j)= 1$ for every $j$, $1\le j\le\lambda_i$,
\end{enumerate}
where $i'=1$ if $i=0$, or $i'=k-1$ if $i=k$.
\label{lem:end-leaf-broadcast}
\end{lemma}

\begin{proof}
Note first that if such a broadcast $\tilde{f}$ exists, then, by Lemma~\ref{lem:stem-0}, 
$\tilde{f}(u)=0$ for every stem $u$ of $CT$.
Therefore, the value of $\sum_{1\le j\le\lambda_i}\tilde{f}(\ell_i^j)$
cannot be strictly less than the value claimed in the lemma since otherwise it would contradict the optimality
of $\tilde{f}$.

By symmetry, it is enough to prove the lemma for the pendent neighbours of $v_0$.
Let $CT_0=CT(\lambda_0,\dots,\lambda_k)$ be a minimal counterexample, with respect
to the subgraph order, to the lemma.
That is, every sub-caterpillar of $CT_0$ satisfies the statement of the lemma
and, for every optimal non-canonical independent broadcast $f$ on $CT_0$ with $\cost(f)>2(\diam(CT)-1)$,
there is a pendent neighbour, say $\ell_0^1$ without loss of generality, of $v_0$ such that
$f(\ell_0^1)=x$ and $x$ is strictly greater than the value claimed by the lemma
(note that, in case 3, if $f(\ell_0^1)=2$ (resp.~0) and $f(\ell_0^2)=0$ (resp.~2),
then we can equivalently assign the value~1 to both of them).
We will prove that such a minimal counterexample cannot exist.

Let $f_0$ be any such independent broadcast on $CT_0$ for which the value $f(\ell_0^1)=x$ is minimal.
We thus have $x\ge 3$ whenever $\lambda_1>0$ or $\lambda_0\ge 3$ (since in this latter case
we can assign value 1 to each of the at least three pendent neighbours of $v_0$, and thus
$x=2$ would imply that $f_0$ is not optimal),
and $x\ge 4$ whenever $\lambda_1=0$.

Since $f_0(\ell_0^1)=x>1$, we have $f_0^*(v_i)=0$ for every $i$, $1\le i\le x-2$,
and $f_0(v_{x-1})=0$.
Moreover, $x-1<k$ since $f_0$ is a non-canonical independent broadcast,
and $v_{x-1}$ cannot be a trunk, since otherwise we could set
$f_0(\ell_0^1)=x+1$ (recall that, by Lemma~\ref{lem:stem-0}, 
$f_0(v_i)=0$ for every stem $v_i$, and thus $f_0(v_x)=0$), 
contradicting the optimality of $f_0$.

Let now $CT_1=(\lambda_{x-1},\dots,\lambda_k)$ be the caterpillar obtained
from $CT_0$ by deleting vertices $v_0,\dots,v_{x-2}$ and their
pendent neighbours (see Figure~\ref{fig:lemme-grappes}(a)). Note that $f_0(u)=0$ for every such deleted vertex $u\neq\ell_0^1$.
Let $f_1$ denote the restriction of $f_0$ to $V(CT_1)$.
Since $f_0(\ell_0^1)=x$, we get
$$f_1(u)=f_0(u)\le\max\{e_{CT_1}(u),d_{CT_0}(u,\ell_0^1)\}\le e_{CT_1}(u)$$
for every vertex $u\in V(CT_1)$, so that $f_1$ is an independent broadcast on $CT_1$
by Observation~\ref{obs:subgraph}.
Moreover, since $\diam(CT_1)=\diam(CT_0)-x+1$, we have
$$\cost(f_1)=\cost(f_0)-x > 2(\diam(CT_0)-1)-x = 2(\diam(CT_1)-1)+x-2.$$
Since $x>1$, we thus have $\cost(f_1)\ge 2(\diam(CT_1)-1)$.
Therefore, since $CT_0$ is a minimal counterexample, we get that 
either $f_1$ is a canonical independent broadcast on $CT_1$ or
there exists an optimal non-canonical
independent broadcast $f'_1$ on $CT_1$ with $\cost(f'_1)\ge\cost(f_1)$
and $f'_1$ satisfies the statement of the lemma.

Suppose first that $f_1$ is a canonical independent broadcast.
This implies 
$$\cost(f_1)=2(\diam(CT_1)-1).$$
Hence, 
$$\cost(f_0)=\cost(f_1)+x=2(\diam(CT_1)-1)+x < 2(\diam(CT_0)-1),$$
which contradicts our assumption on $\cost(f_0)$.

Therefore, there exists an optimal non-canonical
independent broadcast $f'_1$ on $CT_1$ with $\cost(f'_1)\ge\cost(f_1)$
satisfying the statement of the lemma.
If $\cost(f'_1)>\cost(f_1)$, the mapping $f'_0$ 
given by $f'_0(u)=f'_1(u)$ for every vertex $u\in V(CT_1)$
and $f'_0(u)=f_0(u)$ for every vertex $u\in V(CT_0)\setminus V(CT_1)$,
is a non-canonical independent broadcast $f'_0$ on $CT_0$ (since $x\ge 3$)
that contradicts the optimality of $f_0$.

Hence, $f_1$ is optimal and thus satisfies the statement of the lemma.
Let $\tilde{f_1}$ be the non-canonical independent broadcast satisfying items~1 to~5 of the lemma,
and let 
$$m=\max\big\{\tilde{f_1}(\ell_{x-1}^j),\ 1\le j\le\lambda_{x-1}\big\}.$$

We consider two cases, depending on whether $v_{x-2}$ is a stem or not.
Recall that $v_{x-2}\neq v_0$, since $x\ge 3$.

\begin{enumerate}
\item $\lambda_{x-2}>0$.\\
Let $f'_0$ be the non-canonical independent broadcast on $CT_0$ given
by $f'_0(\ell_0^1)=x-1$, $f'_0(\ell_{x-2}^1)=2$,
$f'_0(u)=0$ for every vertex $u\in V(CT_0)\setminus(V(CT_1)\cup\{\ell_0^1,\ell_{x-2}^1\})$,
and either $f'_0(u)=\tilde{f_1}(u)$ for every vertex $u\in V(CT_1)$, if $m\le 2$
(see Figure~\ref{fig:lemme-grappes}(b)),
or $f'_0(\ell_{x-1}^1)=2$ and 
$f'_0(u)=\tilde{f_1}(u)$ for every vertex $u\in V(CT_1)\setminus\{\ell_{x-1}^1\}$, if $m=3$
(see Figure~\ref{fig:lemme-grappes}(c)).
We then get $\cost(f'_0)=\cost(f_0)+1$ if $m\le 2$, contradicting the optimality of $f_0$,
or $\cost(f'_0)=\cost(f_0)$ if $m=3$, 
in which case either $f'_0$ satisfies items~1 to~5 of the lemma
or contradicts the minimality of $x$.

\item $\lambda_{x-2}=0$.\\
If $x=3$, then $\lambda_1=0$ which implies $x\ge 4$, a contradiction.
Hence, we have $x\ge 4$, and thus $v_{x-3}\neq v_0$.
Let $f'_0$ be the non-canonical independent broadcast on $CT_0$ given
by $f'_0(\ell_0^1)=x-2$, $f'_0(\ell_{x-3}^1)=2$,
$f'_0(u)=0$ for every vertex $u\in V(CT_0)\setminus(V(CT_1)\cup\{\ell_0^1,\ell_{x-3}^1\})$,
and $f'_0(u)=\tilde{f_1}(u)$ for every vertex $u\in V(CT_1)$ 
(see Figure~\ref{fig:lemme-grappes}(d)).
We then get $\cost(f'_0)=\cost(f_0)$, and thus
either $f'_0$ satisfies items~1 to~5 of the lemma
or contradicts the minimality of $x$.

\end{enumerate}

This concludes the proof. 
\end{proof}

\begin{figure}
\begin{center}

\begin{tikzpicture}[scale=0.9,x=0.7cm,y=0.7cm]
\CTunplusL{0}{0}{$x$}
\CTzeroplus{2}{0}
\CTunplus{4}{0}{0}
\CTunplus{8}{0}
\node[below] at (0,1) {$v_0$};
\node[below] at (2,1) {$v_1$};
\node[below] at (4,1) {$v_{x-1}$};
\node[below] at (8,1) {$v_k$};
\draw[thick] (0,0) to (2,0);
\draw[thick,dotted] (2,0) to (8,0);
\draw[thick,dashed] (3,1) to (3,-3);
\draw[thick,dashed] (9,1) to (9,-3);
\draw[thick,dashed] (3,1) to (9,1);
\draw[thick,dashed] (3,-3) to (9,-3);

\node at (4.5,-4) {(a) The sub-caterpillar $CT_1$};
\end{tikzpicture}

\mbox{}\\

\begin{tikzpicture}[scale=0.9,x=0.7cm,y=0.7cm]
\CTunplusL{0}{0}{$x$}
\CTunplusL{4}{0}{0}
\CTunplusL{6}{0}{$m$}
\node[below] at (0,1) {$v_0$};
\node[below] at (4,1) {$v_{x-2}$};
\node[below] at (6,1) {$v_{x-1}$};
\draw[thick] (4,0) to (6,0);
\draw[thick,dotted] (0,0) to (4,0);
\draw[thick,dotted] (6,0) to (7,0);
\draw[thick,dashed] (5,1) to (5,-3);
\draw[thick,dashed] (5,1) to (7,1);
\draw[thick,dashed] (5,-3) to (7,-3);

\node at (9,-1) {$\longrightarrow$};

\CTunplusL{11}{0}{$x-1$}
\CTunplusL{15}{0}{2}
\CTunplusL{17}{0}{$m$}
\node[below] at (11,1) {$v_0$};
\node[below] at (15,1) {$v_{x-2}$};
\node[below] at (17,1) {$v_{x-1}$};
\draw[thick] (15,0) to (17,0);
\draw[thick,dotted] (11,0) to (15,0);
\draw[thick,dotted] (17,0) to (18,0);
\draw[thick,dashed] (16,1) to (16,-3);
\draw[thick,dashed] (16,1) to (18,1);
\draw[thick,dashed] (16,-3) to (18,-3);

\node at (9,-4) {(b) $\lambda_{x-2}>0$ and $m\le 2$};
\end{tikzpicture}

\mbox{}\\

\begin{tikzpicture}[scale=0.9,x=0.7cm,y=0.7cm]
\CTunplusL{0}{0}{$x$}
\CTunplusL{4}{0}{0}
\CTunplusL{6}{0}{$3$}
\node[below] at (0,1) {$v_0$};
\node[below] at (4,1) {$v_{x-2}$};
\node[below] at (6,1) {$v_{x-1}$};
\draw[thick] (4,0) to (6,0);
\draw[thick,dotted] (0,0) to (4,0);
\draw[thick,dotted] (6,0) to (7,0);
\draw[thick,dashed] (5,1) to (5,-3);
\draw[thick,dashed] (5,1) to (7,1);
\draw[thick,dashed] (5,-3) to (7,-3);

\node at (9,-1) {$\longrightarrow$};

\CTunplusL{11}{0}{$x-1$}
\CTunplusL{15}{0}{2}
\CTunplusL{17}{0}{$2$}
\node[below] at (11,1) {$v_0$};
\node[below] at (15,1) {$v_{x-2}$};
\node[below] at (17,1) {$v_{x-1}$};
\draw[thick] (15,0) to (17,0);
\draw[thick,dotted] (11,0) to (15,0);
\draw[thick,dotted] (17,0) to (18,0);
\draw[thick,dashed] (16,1) to (16,-3);
\draw[thick,dashed] (16,1) to (18,1);
\draw[thick,dashed] (16,-3) to (18,-3);

\node at (9,-4) {(c) $\lambda_{x-2}>0$ and $m=3$};
\end{tikzpicture}

\mbox{}\\

\begin{tikzpicture}[scale=0.9,x=0.7cm,y=0.7cm]
\CTunplusL{0}{0}{$x$}
\CTunplusL{4}{0}{0}
\CTzero{6}{0}
\CTunplusL{8}{0}{$m$}
\node[below] at (0,1) {$v_0$};
\node[below] at (4,1) {$v_{x-3}$};
\node[below] at (6,1) {$v_{x-2}$};
\node[below] at (8,1) {$v_{x-1}$};
\draw[thick] (4,0) to (8,0);
\draw[thick,dotted] (0,0) to (4,0);
\draw[thick,dotted] (8,0) to (9,0);
\draw[thick,dashed] (7,1) to (7,-3);
\draw[thick,dashed] (7,1) to (9,1);
\draw[thick,dashed] (7,-3) to (9,-3);

\node at (11,-1) {$\longrightarrow$};

\CTunplusL{13}{0}{$x-2$}
\CTunplusL{17}{0}{2}
\CTzero{19}{0}
\CTunplusL{21}{0}{$m$}
\node[below] at (13,1) {$v_0$};
\node[below] at (17,1) {$v_{x-3}$};
\node[below] at (19,1) {$v_{x-2}$};
\node[below] at (21,1) {$v_{x-1}$};
\draw[thick] (17,0) to (21,0);
\draw[thick,dotted] (13,0) to (17,0);
\draw[thick,dotted] (21,0) to (22,0);
\draw[thick,dashed] (20,1) to (20,-3);
\draw[thick,dashed] (20,1) to (22,1);
\draw[thick,dashed] (20,-3) to (22,-3);

\node at (11,-4) {(d) $\lambda_{x-2}=0$};
\end{tikzpicture}

\caption{\label{fig:lemme-grappes}Configurations for the proof of Lemma~\ref{lem:end-leaf-broadcast}}
\end{center}
\end{figure}


We now consider the internal stems of a caterpillar.
Recall that, by Lemma~\ref{lem:stem-0},
$\tilde{f}(v_i)=0$ for every internal stem $v_i$ of $CT$, $1\le i\le k-1$.
The next lemma shows that if $f$ is an optimal non-canonical independent broadcast
on a caterpillar $CT$ with no pair of adjacent trunks, with $\cost(f)>2(\diam(CT)-1)$,
then there exists an optimal non-canonical independent broadcast $\tilde{f}$ on $CT$
such that $\tilde{f}^*(v_i)-\tilde{f}(v_i)=\tilde{f}^*(v_i)>0$ for every internal stem $v_i$ of $CT$, $1\le i\le k-1$.

\begin{lemma}
Let $CT=CT(\lambda_0,\dots,\lambda_k)$ be a caterpillar of length $k\ge 1$,
with no pair of adjacent trunks.
If $f$ is an optimal non-canonical independent broadcast on $CT$ 
with $\cost(f)>2(\diam(CT)-1)$, then there exists an optimal non-canonical 
independent broadcast $\tilde{f}$ on $CT$, 
thus with $\cost(\tilde{f})=\cost(f)$,
such that:
\begin{enumerate}
\item $\tilde{f}$ satisfies the five items of Lemma~\ref{lem:end-leaf-broadcast},
\item for every $i$, $1\le i\le k-1$, if $\lambda_i>0$,
then $\tilde{f}^*(v_i)>0$.
\end{enumerate} 
\label{lem:internal-stems}
\end{lemma}

\begin{proof}
We know by Lemma~\ref{lem:end-leaf-broadcast} that there exists 
an optimal non-canonical independent broadcast $\tilde{f}$ on $CT$,
with $\cost(\tilde{f})=\cost(f)$, 
satisfying the five items of Lemma~\ref{lem:end-leaf-broadcast}.
Moreover, one suppose that $\tilde{f}$ has been chosen in such a way
that $V^+_{\tilde{f}}$ contains the largest possible number of pendent vertices.

Suppose to the contrary that there exists a vertex $v_i$, $1\le i\le k-1$,
with $\lambda_i>0$ and $\tilde{f}^*(v_i)=0$, and that for every $j<i$,
$\tilde{f}^*(v_j)>0$ whenever $\lambda_j>0$.
We consider three cases. 

\begin{enumerate}
\item $i=1$ or $i=k-1$.\\
By symmetry, it suffices to consider the case $i=1$.
By Lemma~\ref{lem:end-leaf-broadcast}, we know that $\tilde{f}(\ell_0^j)\le 2$
for every $j$, $1\le j\le\lambda_0$.
Therefore, no pendent neighbour of $v_1$ is $\tilde{f}$-dominated by a pendent neighbour
of $v_0$. Let $y$ be the vertex of $CT$ that $\tilde{f}$-dominates the pendent neighbours of $v_1$
(note that $y$ is necessarily unique),
and $g$ be the mapping defined as follows.
For every vertex $u$ of $CT$, let
$$g(u)=\left\lbrace\begin{array}{ll}
\tilde{f}(y)-1 & \mbox{if $u=y$,}\\
1 & \mbox{if $u=\ell^1_1$,}\\
1 & \mbox{if $u\neq\ell^1_1$, $u$ is $\tilde{f}$-dominated only by $y$ and $d_{CT}(u,y)=\tilde{f}(y)$,}\\
\tilde{f}(u) & \mbox{otherwise.}\\
\end{array}\right.
$$
We claim that the mapping $g$ is  a non-canonical independent broadcast on $CT$
with $\cost(g)\ge\cost(\tilde{f})$.
Indeed, all vertices $x$ with $d_{CT}(x,y)<\tilde{f}(y)$ that were $\tilde{f}$-dominated by $y$
are still $g$-dominated by $y$,
and all vertices $x'\neq\ell^1_1$ with $d_{CT}(x',y)=\tilde{f}(y)$ that were $\tilde{f}$-dominated only by $y$
are now $g$-broadcast vertices with $g(x')=1$
(note that since every such $x'$ was $\tilde{f}$-dominated only by $y$,
we have $g(z)=\tilde{f}(z)=0$ for every neighbour $z$ of $x'$).

Now, if there exists a vertex $z$ which is $\tilde{f}$-dominated only by $y$,
we get $\cost(g)\ge\cost(\tilde{f})+1$, contradicting the optimality of $\tilde{f}$.
If no such vertex exists, we get $\cost(g)=\cost(\tilde{f})$ and 
$V^+_g$ contains more pendent vertices than $V^+_{\tilde{f}}$,
contrary to our assumption.


\item $i=2$ and $\lambda_1=0$, or $i=k-2$ and $\lambda_{k-1}=0$.\\
By symmetry, it suffices to consider the case $i=2$.
By Lemma~\ref{lem:end-leaf-broadcast}, we know that $\tilde{f}(\ell_0^j)\le 3$
for every $j$, $1\le j\le\lambda_0$.
Therefore, no pendent neighbour of $v_2$ is $\tilde{f}$-dominated by a pendent neighbour
of $v_0$. Let $y$ be the (unique) vertex of $CT$ that $\tilde{f}$-dominates the pendent neighbours of $v_2$
(note that we necessarily have $\tilde{f}(y)\ge 2$).

If $y=v_3$ and $\tilde{f}(v_3)=3$ (since $\tilde{f}^*(v_0)>0$, we necessarily
have $\tilde{f}(v_3)\le 3$), 
we define the mapping $g$ as follows.
For every vertex $u$ of $CT$, let
$$g(u)=\left\lbrace\begin{array}{ll}
0 & \mbox{if $u=v_3$,}\\
3 & \mbox{if $u=\ell^1_2$,}\\
1 & \mbox{if $u\neq\ell^1_2$, $u$ is $\tilde{f}$-dominated only by $v_3$ and $d_{CT}(u,y)=2$,}\\
\tilde{f}(u) & \mbox{otherwise.}\\
\end{array}\right.
$$

Otherwise (including the case $y=v_3$ and $\tilde{f}(v_3)=2$), the mapping $g$ is defined by
$$g(u)=\left\lbrace\begin{array}{ll}
\tilde{f}(y)-2 & \mbox{if $u=y$,}\\
2 & \mbox{if $u=\ell^1_2$,}\\
1 & \mbox{if $u\neq\ell^1_2$, $u$ is $\tilde{f}$-dominated only by $y$ and $d_{CT}(u,y)=\tilde{f}(y)-1$,}\\
\tilde{f}(u) & \mbox{otherwise,}\\
\end{array}\right.
$$
for every vertex $u$ of $CT$.

In both cases, the mapping $g$ is again a non-canonical independent broadcast on $CT$
with $\cost(g)\ge\cost(\tilde{f})$.
Indeed, all vertices $x$ with $d_{CT}(x,y)<\tilde{f}(y)-1$ that were $\tilde{f}$-dominated by $y$
are $g$-dominated by $\ell^2_1$ (if $y=v_3$)
or still $g$-dominated by $y$ (if $y\neq v_3$),
and all vertices $x'\neq\ell^1_2$ with $\tilde{f}(y)-1\le d_{CT}(x',y)\le \tilde{f}(y)$ that were $\tilde{f}$-dominated only by $y$
are now either $g$-broadcast vertices (if $d_{CT}(x',y)=\tilde{f}(y)-1$) or $g$-dominated
by a vertex $x''$ with $d_{CT}(x'',y)=\tilde{f}(y)-1$ and $g(x'')=1$.

We then get a contradiction as in Case~1.

\item $2<i<k-2$, or $i=2$ and $\lambda_1>0$, or $i=k-2$ and $\lambda_{k-1}>0$.\\
In this case, we have $\tilde{f}^*(v_j)>0$ for every vertex $v_j$ with $j<i$ and $\lambda_j>0$.
Note also that we have at least two such vertices $v_j$ with $j<i$ and $\lambda_j>0$.

By symmetry, it suffices to consider the cases $2<i<k-2$, and $i=2$ (with $\lambda_1>0$).
We consider three subcases.

\begin{enumerate}
\item 
Suppose first that the pendent neighbours of $v_i$ are $\tilde{f}$-dominated
only by a vertex $y=v_{j_0}$ or $y=\ell_{j_0}^{k_0}$ with $j_0<i$ and $1\le {k_0}\le\lambda_{j_0}$.
Observe that the pendent neighbours of $v_i$ cannot be $\tilde{f}$-dominated
by two such vertices, say $y$ and $y'$, since we would have $d_{CT}(y,y')<d_{CT}(y,\ell_i^1)$
so that $\tilde{f}$ would not be independent.
Since $\tilde{f}^*(v_j)>0$ for every $j<i$ such that $\lambda_j>0$, 
we necessarily have, by Lemma~\ref{lem:stem-0}, either 
$y$ is a pendent neighbour of $v_{i-1}$, if $\lambda_{i-1}>1$,
or a pendent neighbour of $v_{i-2}$, if $\lambda_{i-1}=0$.
Moreover, since $\tilde{f}^*(v_j)>0$ for every $j<i$ such that $\lambda_j>0$,
and since we have at least two such vertices,
we necessarily have $\tilde{f}(y)\le 3$.
This implies in particular $\lambda_{i-1}>0$, as otherwise we would have 
$\tilde{f}(y)\le 3$ and $d_{CT}(y,\ell_i^1)=4$, contradicting the fact that
$y$ $\tilde{f}$-dominates $\ell_i^1$, and thus $y$ is a pendent neighbour of $v_{i-1}$.

Let now $g$ be the mapping defined as follows.
For every vertex $u$ of $CT$, let
$$g(u)=\left\lbrace\begin{array}{ll}
\tilde{f}(y)-1 & \mbox{if $u=y$,}\\
1 & \mbox{if $u=\ell_i^1$,}\\
1 & \mbox{if $u\neq\ell_i^1$, $u$ is $\tilde{f}$-dominated only by $y$ and $d_{CT}(u,y)=\tilde{f}(y)$,}\\
\tilde{f}(u) & \mbox{otherwise.}\\
\end{array}\right.
$$

Again, the mapping $g$ is a non-canonical independent broadcast on $CT$
with $\cost(g)\ge\cost(\tilde{f})$.
Indeed, all vertices $x$ with $d_{CT}(x,y)<\tilde{f}(y)$ that were $\tilde{f}$-dominated by $y$
are still $g$-dominated either by $y$,
and all vertices $x'\neq\ell_i^1$ with $d_{CT}(x',y)=\tilde{f}(y)$ that were $\tilde{f}$-dominated only by $y$
are now $g$-broadcast vertices.

We then get a contradiction as in Cases 1 and~2.

\item Suppose now that the pendent neighbours of $v_i$ are $\tilde{f}$-dominated
only by a vertex $y=v_{j_0}$ (with $\lambda_{j_0}=0$)
 or $y=\ell_{j_0}^{k_0}$ ($1\le {k_0}\le\lambda_{j_0}$), with $j_0>i$.
Observe that, using the same argument as in Case~(a), such a vertex $y$ must be unique.
%
%
Moreover, we necessarily have $\tilde{f}(y)\ge 2$.

If $\lambda_{i-1}=0$, we consider two cases, as we did in Case 2.
If $y=v_{i+1}$ and $\tilde{f}(v_{i+1})=3$, 
we define the mapping $g$ by
$$g(u)=\left\lbrace\begin{array}{ll}
0 & \mbox{if $u=v_{i+1}$,}\\
3 & \mbox{if $u=\ell_i^1$,}\\
1 & \mbox{if $u\neq\ell_i^1$, $u$ is $\tilde{f}$-dominated only by $y$ and $d_{CT}(u,y)=2$,}\\
\tilde{f}(u) & \mbox{otherwise,}\\
\end{array}\right.
$$
for every vertex $u$ of $CT$. Otherwise, the mapping $g$ is defined by
$$g(u)=\left\lbrace\begin{array}{ll}
\tilde{f}(y)-2 & \mbox{if $u=y$,}\\
2 & \mbox{if $u=\ell_i^1$,}\\
1 & \mbox{if $u\neq\ell_i^1$, $u$ is $\tilde{f}$-dominated only by $y$ and $d_{CT}(u,y)=\tilde{f}(y)-1$,}\\
\tilde{f}(u) & \mbox{otherwise,}\\
\end{array}\right.
$$
for every vertex $u$ of $CT$. 

Otherwise, that is, $\lambda_{i-1}>0$,
we define the mapping $g$ as follows.
For every vertex $u$ of $CT$, let
$$g(u)=\left\lbrace\begin{array}{ll}
\tilde{f}(y)-1 & \mbox{if $u=y$,}\\
1 & \mbox{if $u=\ell_i^1$,}\\
1 & \mbox{if $u\neq\ell_i^1$, $u$ is $\tilde{f}$-dominated only by $y$ and $d_{CT}(u,y)=\tilde{f}(y)$,}\\
\tilde{f}(u) & \mbox{otherwise.}\\
\end{array}\right.
$$

Again, using similar arguments,
in each case the above-defined mapping is a non-canonical independent broadcast on $CT$
with $\cost(g)\ge\cost(\tilde{f})$ and the contradiction arises as in Cases 1 and~2.

\item Suppose finally that the pendent neighbours of $v_i$ are $\tilde{f}$-dominated
both by a vertex $y_1=v_{j_1}$ or $y_1=\ell_{j_1}^{k_1}$ with $j_1<i$ and $1\le k_1\le\lambda_{j_1}$,
and
by a vertex $y_2=v_{j_2}$ or $y_2=\ell_{j_2}^{k_2}$ with $j_2>i$ and $1\le k_2\le\lambda_{j_2}$
(again, both $y_1$ and $y_2$ must be unique).
In that case, as discussed in Case~(a) above, we necessarily have $\lambda_{i-1}>0$.
Moreover, we necessarily have $\tilde{f}(y_1)= 3$ and $\tilde{f}(y_2)\ge 2$.

Let now $g$ be the mapping defined as follows.
For every vertex $u$ of $CT$, let
$$g(u)=\left\lbrace\begin{array}{ll}
\tilde{f}(y_1)-1 & \mbox{if $u=y_1$,}\\
\tilde{f}(y_2)-1 & \mbox{if $u=y_2$,}\\
2 & \mbox{if $u=\ell_i^1$,}\\
1 & \mbox{if $u\neq\ell_i^1$, $u$ is $\tilde{f}$-dominated only by $y_2$ and $d_{CT}(u,y_2)=\tilde{f}(y_2)$,}\\
\tilde{f}(u) & \mbox{otherwise.}\\
\end{array}\right.
$$

Note here that no vertex at distance $\tilde{f}(y_1)$ from $y_1$ can be $\tilde{f}$-dominated only by $y_1$.
Indeed, suppose that such a vertex, say $w$, exists.
Clearly, $w$ cannot be ``to the left of $v_i$''
since this would imply  $w=v_{i-3}$ and $\lambda_{i-2}=0$, but
in that case $w$ is also $\tilde{f}$-dominated by at least one of its pendent neighbours.
On the other hand, $w$ cannot be ``to the right of $v_i$''
since in that case $w$ would also be $\tilde{f}$-dominated by $y_2$.

Again, using similar arguments,
the above-defined mapping is a non-canonical independent broadcast on $CT$
with $\cost(g)\ge\cost(\tilde{f})$ and the contradiction arises as in Cases 1 and~2.

\end{enumerate}

\end{enumerate}

We thus get a contradiction in each case. This completes the proof.
\end{proof}



Our aim now is to prove that 
if $f$ is an optimal non-canonical independent broadcast
on a caterpillar $CT$ with no pair of adjacent trunks, with $\cost(f)>2(\diam(CT)-1)$,
then $\cost(f)=\cost(\beta^*)$ (Lemma~\ref{lem:beta-star-is-the-best} below).
We first prove that for every such broadcast $f$, 
$f(v_i)\le 1$ for every trunk $v_i$. 
This easily follows from 
Lemma~\ref{lem:internal-stems}.

\begin{lemma}
Let $CT=CT(\lambda_0,\dots,\lambda_k)$ be a caterpillar of length $k\ge 1$,
with no pair of adjacent trunks.
If $f$ is an optimal non-canonical independent broadcast on $CT$ 
with $\cost(f)>2(\diam(CT)-1)$, then there exists an optimal non-canonical 
independent broadcast $\tilde{f}$ on $CT$, 
thus with $\cost(\tilde{f})=\cost(f)$,
such that:
\begin{enumerate}
\item $\tilde{f}$ satisfies the two items of Lemma~\ref{lem:internal-stems},
\item for every $i$, $1\le i\le k-1$, if $\lambda_i=0$,
then $\tilde{f}^*(v_i)\le 1$.
\end{enumerate} 
\label{lem:internal-trunks}
\end{lemma}

\begin{proof}
We know by Lemma~\ref{lem:internal-stems} that
there exists an optimal non-canonical independent broadcast $\tilde{f}$ on $CT$ 
satisfying the two items of Lemma~\ref{lem:internal-stems},
so that, in particular, $\tilde{f}^*(v_j)>0$ for every stem $v_j$, $0\le j\le k$.
Since $CT$ has no pair of adjacent trunks, and $\tilde{f}$ is independent,
we thus necessarily have $\tilde{f}^*(v_i)\le 1$ for every trunk $v_i$, $1\le i\le k-1$.
\end{proof}



Finally, the next lemma will show that the cost of any optimal non-canonical independent broadcast on
a caterpillar $CT$ of length $k\ge 1$ with no pair of adjacent trunks cannot exceed the value $\beta^*(CT)$.

We first introduce more notation.
Let $CT$ be a caterpillar of length $k\ge 1$, with no pair of adjacent trunks.
We denote by $\sigma$ a sequence of $\ell$ consecutive spine vertices in $CT$,
that is, $\sigma=v_i\dots v_{i+\ell-1}$, with $\ell\le k+1$ and $0\le i\le k-\ell+1$.
For such a given sequence $\sigma=v_i\dots v_{i+\ell-1}$, we denote by $t_\sigma$ the number of
trunks in $\sigma$, that is,
$$t_\sigma=\left|\{v_j\ |\ i\le j\le i+\ell-1\ \mbox{and}\ \lambda_j=0\}\right|.$$
If $f$ is an independent broadcast on $CT$, we then denote by $f^*(\sigma)$
the {\it weight of $\sigma$}, that is, 
$$f^*(\sigma)=\sum_{0\le j\le\ell-1}f^*(v_{i+j}).$$


\begin{lemma}
Let $CT=CT(\lambda_0,\dots,\lambda_k)$ be a caterpillar of length $k\ge 1$, with no pair of adjacent trunks,
and $f$ be an optimal non-canonical independent broadcast on $CT$ with $\cost(f)>2(\diam(CT)-1)$.
Then there exists an optimal non-canonical independent broadcast $\tilde{f}$ on $CT$,
thus with $\cost(\tilde{f})=\cost(f)$, such that:
\begin{enumerate}
\item $\tilde{f}$ satisfies the two items of Lemma~\ref{lem:internal-trunks}.
\item For every $i$, $0\le i\le k$, if $\lambda_i\ge 3$, then $\tilde{f}^*(v_i)\le\lambda_i$.
\item If $v_av_{a+1}$, $0\le a<k$, is an occurrence of the pattern $1^+2^-$ 
(resp. of the pattern $2^-1^+$), 
then $\tilde{f}^*(v_{a+1})\le 2$ (resp. $\tilde{f}^*(v_{a})\le 2$).
\item If $v_a\sigma v_b$ is an occurrence of the pattern $1^+2^-(02^-)^{+r}1^+$,
then $\tilde{f}^*(\sigma)\le 3t_\sigma +2$ 
if $v_a\sigma v_b$ is an occurrence of the pattern $1^+2(02)^{+r}1^+$,
and $\tilde{f}^*(\sigma)\le 3t_\sigma +1$ otherwise.
\item If $\sigma$ is an occurrence of the pattern $02^-(02^-)^{*r}0$,
then $\tilde{f}^*(\sigma)\le 3t_\sigma -2$ 
if $v_a\sigma v_b$ is an occurrence of the pattern $02(02)^{*r}0$,
and $\tilde{f}^*(\sigma)\le 3t_\sigma -3$ otherwise.
\item If $\sigma$ is an occurrence of the pattern $[2^-(02^-)^{*r}0$
or of the pattern $02^-(02^-)^{*r}]$,
then $\tilde{f}^*(\sigma)\le 3t_\sigma$.
\end{enumerate}
\label{lem:not-greater-than-beta-star}
\end{lemma}

\begin{proof}
We consider the six items of the lemma.
\begin{enumerate}
\item We know by Lemma~\ref{lem:internal-trunks} that
there exists an optimal non-canonical independent broadcast $\tilde{f}$ on $CT$ 
satisfying the two items of Lemma~\ref{lem:internal-trunks},
so that, in particular, $\tilde{f}^*(v_i)>0$ for every stem $v_i$, $0\le i\le k$
and $\tilde{f}^*(v_j)\le 1$ for every trunk $v_j$, $1\le j\le k-1$.
We thus assume for all following items that such 
an optimal non-canonical independent broadcast $\tilde{f}$ on $CT$ 
has been chosen.

\item 
Suppose to the contrary that there exists some $i$, $0\le i\le k$,
with $\tilde{f}^*(v_i)>\lambda_i\ge 3$.
This implies that $v_i$ has exactly one pendent neighbour,
say $\ell_i^1$ without loss of generality, which is an
$\tilde{f}$-broadcast vertex.
Since $\tilde{f}(\ell_i^1)\ge 4$, we necessarily have a stem
$v$ with $d_{CT}(v_i,v)\le 2$ 
and $\tilde{f}^*(v)=0$,
contradicting our assumption that $\tilde{f}$ satisfies Lemma~\ref{lem:internal-stems}.

\item 
Let $v_av_{a+1}$, $0\le a<k$,  be an occurrence of the pattern $1^+2^-$
(the case $2^-1^+$ is similar, by symmetry).
By Lemmas \ref{lem:stem-0} and~\ref{lem:internal-stems},
we know that $\tilde{f}^*(v_a)>0$ and $\tilde{f}(v_a)=0$.
This clearly implies $\tilde{f}^*(v_{a+1})\le 2$.

\item 
Let $v_a\sigma v_b=v_iv_{i+1}\dots v_{i+2r+2}$ be
an occurrence of the pattern $1^+2(02)^{+r}1^+$,
for some $i$, $0\le i\le k-2r-2$.
We thus have $t_\sigma=r$.
Since $\tilde{f}$ satisfies Lemma~\ref{lem:internal-trunks},
we have $\tilde{f}^*(v_i)>0$, $\tilde{f}^*(v_{i+2r+2})>0$,
$\tilde{f}^*(v_{i+2j+1})>0$ for every $j$, $0\le j\le r$,
and $\tilde{f}^*(v_{i+2j})\le 1$ for every $j$, $1\le j\le r$.
This implies
\begin{equation}
\tilde{f}^*(v_{i+1})\le 2,\ \tilde{f}^*(v_{i+2r+1})\le 2,\ 
\mbox{and}\ \tilde{f}^*(v_{i+2j+1})\le 3\ \mbox{ for every}\ j,\ 1\le j\le r-1.
\label{eq:cas3}\end{equation}


We consider three subcases, according to the number of trunks in $\sigma$ that
are broadcast vertices. 

\begin{enumerate}
\item $\tilde{f}(v_{i+2j})=1$ for every $j$, $1\le j\le r$.\\
In that case, every pendent vertex in $\sigma$ is an
$\tilde{f}$-broadcast vertex, with $\tilde{f}$-value 1. 
This gives
$$\tilde{f}^*(\sigma)=\lambda(\sigma)+\tau(\sigma)\le 2(r+1)+r=3r+2=3t_\sigma+2,$$
if $v_a\sigma v_b$ is an occurrence of the pattern $1^+2(02)^{+r}1^+$,
and
$$\tilde{f}^*(\sigma)=\lambda(\sigma)+\tau(\sigma)\le 1+2r+r=3r+1=3t_\sigma+1,$$
otherwise (since we have at least one stem in $\sigma$ with $\tilde{f}$-value 1).

\item $\tilde{f}(v_{i+2j})=0$ for every $j$, $1\le j\le r$.\\
In that case, by~(\ref{eq:cas3}), we get
$$\tilde{f}^*(\sigma)\le 2 + 3(r-1) + 2=3r+1=3t_\sigma+1.$$

\item Not all trunks in $\sigma$ have the same $\tilde{f}$-value.\\
Suppose that $\tilde{f}$ has been chosen in such a way
that the number of trunks in $\sigma$ with $\tilde{f}$-value 0 is maximal.
In that case, $\sigma$ contains two consecutive trunks, say $v_{i+2j_0}$
and $v_{i+2j_0+2}$, $1\le j_0\le r-1$, 
with $\tilde{f}(v_{i+2j_0})=0$ and $\tilde{f}(v_{i+2j_0+2})=1$, without
loss of generality (by symmetry).
This implies $\tilde{f}^*(v_{i+2j_0+1})=\lambda_{i+2j_0+1}\le 2$.
We can then modify $\tilde{f}$ by setting
$\tilde{f}(v_{i+2j_0})=\tilde{f}(v_{i+2j_0+2})=0$,
$\tilde{f}(\ell_{i+2j_0+1}^1)=3$ (and $\tilde{f}(\ell_{i+2j_0+1}^2)=0$ if $\lambda_{i+2j_0+1}=2$),
contradicting our assumption on the maximality of the number of trunks with
$\tilde{f}$-value 0.
Therefore, this case cannot occur and we are done.
\end{enumerate}

\item 
The proof uses the same ideas as the proof of the previous case.
Let $\sigma=v_iv_{i+1}\dots v_{i+2r+2}$ be
an occurrence of the pattern $02^-(02^-)^{*r}0$,
for some $i$, $1\le i\le k-2r-3$.
We thus have $t_\sigma=r+2$.
Since $\tilde{f}$ satisfies Lemma~\ref{lem:internal-trunks},
we have 
\begin{equation}
0<\tilde{f}^*(v_{i+2j+1})\le 3\ \mbox{ for every}\ j,\ 0\le j\le r,
\label{eq:cas4-1}\end{equation}
and 
\begin{equation}
\tilde{f}^*(v_{i+2j})\le 1\ \mbox{ for every}\ j,\ 0\le j\le r+1.
\label{eq:cas4-2}\end{equation}


We consider three subcases, according to the number of trunks in $\sigma$ that
are broadcast vertices. 

\begin{enumerate}
\item $\tilde{f}(v_{i+2j})=1$ for every $j$, $0\le j\le r+1$.\\
In that case, every pendent vertex in $\sigma$ is an
$\tilde{f}$-broadcast vertex, with $\tilde{f}$-value 1. 
This gives
$$\tilde{f}^*(\sigma)=\lambda(\sigma)+\tau(\sigma)\le 2(r+1)+r+2=3r+4=3t_\sigma-2,$$
if $\sigma$ is an occurrence of the pattern $02(02)^{*r}0$,
and
$$\tilde{f}^*(\sigma)=\lambda(\sigma)+\tau(\sigma)\le 1+2r+r+2=3r+3=3t_\sigma-3,$$
otherwise (since we have at least one stem in $\sigma$ with $\tilde{f}$-value 1).

\item $\tilde{f}(v_{i+2j})=0$ for every $j$, $0\le j\le r+1$.\\
In that case, by (\ref{eq:cas4-1}) and~(\ref{eq:cas4-2}), we get
$$\tilde{f}^*(\sigma)\le 3(r+1)=3r+3=3t_\sigma-3.$$

\item Not all trunks in $\sigma$ have the same $\tilde{f}$-value.\\
Suppose that $\tilde{f}$ has been chosen in such a way
that the number of trunks in $\sigma$ with $\tilde{f}$-value 0 is maximal.
In that case, $\sigma$ contains two consecutive trunks, say $v_{i+2j_0}$
and $v_{i+2j_0+2}$, $0\le j_0\le r$, 
with $\tilde{f}(v_{i+2j_0})=0$ and $\tilde{f}(v_{i+2j_0+2})=1$, without
loss of generality (by symmetry).
This implies $\tilde{f}^*(v_{i+2j_0+1})=\lambda_{i+2j_0+1}\le 2$.
We can then modify $\tilde{f}$ by setting
$\tilde{f}(v_{i+2j_0})=\tilde{f}(v_{i+2j_0+2})=0$,
$\tilde{f}(\ell_{i+2j_0+1}^1)=3$ (and $\tilde{f}(\ell_{i+2j_0+1}^2)=0$ if $\lambda_{i+2j_0+1}=2$),
contradicting our assumption on the maximality of the number of trunks with
$\tilde{f}$-value 0.
Therefore, this case cannot occur and we are done.
\end{enumerate}

\item 
Let $v_0\dots v_{2r+1}$ be an occurrence of the pattern $[2^-(02^-)^{*r}0$
(the case $02^-(02^-)^{*r}]$ is similar, by symmetry).
We first prove that for every $i$, $0\le i\le r$,
$\tilde{f}^*(v_{2i})+\tilde{f}^*(v_{2i+1})\le 3$.
By Lemma~\ref{lem:internal-trunks}, we know that $\tilde{f}(v_{2i+1})\le 1$.
If $\tilde{f}(v_{2i+1})=1$, we then have $\tilde{f}(\ell_{2i}^j)\le 1$ for 
every pendent neighbour $\ell_{2i}^j$
of $v_{2i}$, and thus $\tilde{f}^*(v_{2i})\le\lambda_{2i}\le 2$.
On the other hand, if $\tilde{f}(v_{2i+1})=0$, we have $\tilde{f}^*(v_{2i})\le 3$
(which implies $\tilde{f}(\ell_{2i}^j)=3$ for a pendent neighbour $\ell_{2i}^j$
of $v_{2i}$)
since otherwise we would have $\tilde{f}^*(v_{2i+2})=0$,
contradicting Lemma~\ref{lem:internal-stems}.
In both cases, we thus get the desired inequality.

Since $\sigma$ contains exactly $r+1=t_\sigma$ distinct
pairs of vertices of the form $(v_{2i},v_{2i+1})$, we get
$$\tilde{f}^*(\sigma)=\sum_{i=0}^{i=r}\left(\tilde{f}^*(v_{2i})+\tilde{f}^*(v_{2i+1})\right)\le 3(r+1)=3t_\sigma.$$

\end{enumerate}
This completes the proof.
\end{proof}

The following lemma states that Lemma~\ref{lem:not-greater-than-beta-star}
covers all possible caterpillars that admit a non-canonical independent broadcast
with sufficiently large cost.

\begin{lemma}
If $CT=CT(\lambda_0,\dots,\lambda_k)$ is a caterpillar of length $k\ge 1$, with no pair of adjacent trunks, such that there exists
an optimal non-canonical independent broadcast $f$ on $CT$ with $\cost(f)>2(\diam(CT)-1)$,
then Lemma~\ref{lem:not-greater-than-beta-star} gives an upper bound on $\cost(f)$.
\label{lem:all-but-special-case}
\end{lemma}

\begin{proof}
Let $CT=CT(\lambda_0,\dots,\lambda_k)$ be a caterpillar of length $k\ge 1$, with no pair of adjacent trunks,
$f$ be an optimal non-canonical independent broadcast on $CT$ with $\cost(f)>2(\diam(CT)-1)$,
and $v_i$, $0\le i\le k$, a spine vertex of $CT$.
 
If $\lambda_i\ge 3$, then $f^*(v_i)=\lambda_i$ by item~5 of Lemma~\ref{lem:end-leaf-broadcast},
and thus by item~1 of Lemma~\ref{lem:not-greater-than-beta-star}.

If $\lambda_i=0$, then $f^*(v_i)\le 1$ by item~2 of Lemma~\ref{lem:internal-trunks},
and thus by item~1 of Lemma~\ref{lem:not-greater-than-beta-star}.

Suppose now that $1\le \lambda_i\le 2$.
If $i=0$ or $i=k$, then $f^*(v_i)\le 3$ by items 1 to~4 of Lemma~\ref{lem:end-leaf-broadcast},
and thus by item~1 of Lemma~\ref{lem:not-greater-than-beta-star}.
We assume now that $1\le i\le k-1$.
If $\lambda_{i-1}>0$ or $\lambda_{i+1}>0$,
then $f^*(v_i)\le 2$ by item~3 of Lemma~\ref{lem:not-greater-than-beta-star}.

The remaining case is thus $1\le i\le k-1$, $\lambda_{i-1}=0$ and $\lambda_{i+1}=0$.
We consider the set of all occurrences of a pattern, in which
$0$'s and $2^-$'s alternate, that contain vertices $v_{i-1}$, $v_i$ and $v_{i+1}$.
Let $\sigma=v_av_{a+1}\dots v_b$, $0\le a\le i-1 < i+1\le b\le k$ be such an occurrence
with maximal length.
Note here that we necessarily have $v_a\neq v_i$ and $v_b\neq v_i$.
We consider three cases.

\begin{enumerate}
\item $\lambda_a=\lambda_b=0$.\\
By the maximality of $\sigma$, we necessarily have
$\lambda_{a-1}\ge 3$ and $\lambda_{b+1}\ge 3$.
Therefore, the value of $f^*(\sigma)$ is bounded by item~5 of 
Lemma~\ref{lem:not-greater-than-beta-star}.

\item $\lambda_a=0$ and $\lambda_b>0$ 
(the case $\lambda_a>0$ and $\lambda_b=0$ is similar, by symmetry).\\
By the maximality of $\sigma$, we necessarily have
$\lambda_{a-1}\ge 3$ and either $b=k$, or $b<k$ and $\lambda_{b+1}\ge 1$.
If $b=k$, then the value of $f^*(\sigma)$ is bounded by item~6 of 
Lemma~\ref{lem:not-greater-than-beta-star}.
If $b<k$ and $\lambda_{b+1}\ge 1$,
then $f^*(v_a\dots v_{b-1})$ is bounded by item~5 of 
Lemma~\ref{lem:not-greater-than-beta-star}.

\item $\lambda_a>0$ and $\lambda_b>0$.\\
By the maximality of $\sigma$, we necessarily have
(i) either $a=0$, or $a>0$ and $\lambda_{a-1}\ge 1$, and 
(ii) either $b=k$, or $b<k$ and $\lambda_{b+1}\ge 1$.

If $a>0$ and $b<k$, then the value of $f^*(\sigma)$ is bounded by item~4 of 
Lemma~\ref{lem:not-greater-than-beta-star}.

If $a=0$ and $b<k$ (the case $a>0$ and $b=k$ is similar, by symmetry),
then the value of $f^*(v_a\dots v_{b-1})$ is bounded by item~6 of 
Lemma~\ref{lem:not-greater-than-beta-star}.

Finally, if $a=0$ and $b=k$, the caterpillar $CT$ has 
pattern $2^-(02^-)^{+r}$.
In that case, we have $\diam(CT)=2r+2$ and thus $2(\diam(CT)-1)=4r+2$.
But by Lemmas \ref{lem:internal-stems} and~\ref{lem:internal-trunks}
(as discussed in the proof of item~6 of Lemma~\ref{lem:not-greater-than-beta-star}),
we have $f^*(v_j)+f^*(v_{j+1})\le 3$ for every $j$, $0\le j\le 2r-2$.
Moreover, by item~2 of Lemma~\ref{lem:end-leaf-broadcast},
we have $f^*(v_{2r})=3$.
Therefore, $f^*(CT)\le 3r+3\le 4r+2=2(\diam(CT)-1)$.
This contradicts our assumption on the value of $\cost(f)$, and thus
this case cannot occur.

\end{enumerate}

Therefore, in all cases, either $f^*(v_i)$ or $f^*(\sigma)$ for an occurrence
$\sigma$ of a pattern containing $v_i$ is bounded by some item of
Lemma~\ref{lem:not-greater-than-beta-star}.
This concludes the proof.
\end{proof}

Using Lemmas \ref{lem:not-greater-than-beta-star} and~\ref{lem:all-but-special-case}, 
we can now prove that
no optimal non-canonical independent broadcast $f$ on $CT$ with $\cost(f)>2(\diam(CT)-1)$
and $\cost(f)>\beta^*(CT)$ exists.

\begin{lemma}
Let $CT=CT(\lambda_0,\dots,\lambda_k)$ be a caterpillar of length $k\ge 1$, with no pair of adjacent trunks,
and $f$ be an optimal non-canonical independent broadcast on $CT$ with $\cost(f)>2(\diam(CT)-1)$.
We then have $\cost(f)\le\beta^*(CT)$.
\label{lem:beta-star-is-the-best}
\end{lemma}

\begin{proof}
Let us denote by $f_4$ the non-canonical independent broadcast on $CT$
constructed in the proof of Lemma~\ref{lem:beta-star}, thus with $\cost(f_4)=\beta^*(CT)$.
%
By considering the four steps involved in the construction of $f_4$, it 
clearly appears that $f_4$ satisfies the five items of Lemma~\ref{lem:end-leaf-broadcast},
item~2 of Lemma~\ref{lem:internal-stems} and item~2 of Lemma~\ref{lem:internal-trunks}.
Therefore, $f_4$ satisfies item~1 of Lemma~\ref{lem:not-greater-than-beta-star}.
Moreover, if $v_i$ is a trunk that does not appear in any pattern
considered in Lemma~\ref{lem:not-greater-than-beta-star},
then $f_4(v_i)=1$. Indeed, the $f_4$-value of
$v_i$ is set to~1 in step~1 of Lemma~\ref{lem:beta-star}
and is not modified in steps 2 to~4.

We now prove that $f_4$ satisfies the five last items
of Lemma~\ref{lem:not-greater-than-beta-star} and that, in each case,
the upper bound is attained. We will refer to steps 1 to~4 of
the proof of Lemma~\ref{lem:beta-star} and to the corresponding
intermediate independent broadcasts $f_1$ to~$f_3$.
Recall first that in step~1, every trunk and every pendent vertex is assigned the value~1.

\begin{enumerate}
\item {\it Item~2 of Lemma~\ref{lem:not-greater-than-beta-star}}.\\
If $v_i$ is a stem with $\lambda_i\ge 3$, the value of
its pendent neighbours is not modified in steps 2 to~4.
Therefore, we get $f_4^*(v_i)=f_1^*(v_i)=\lambda_i$ for every such $v_i$.

\item {\it Item~3 of Lemma~\ref{lem:not-greater-than-beta-star}}.\\
Let $v_av_{a+1}$, $0\le a<k$,  be an occurrence of the pattern $1^+2^-$
(the case $2^-1^+$ is similar, by symmetry).
Note here that if $v_{a+1}$ is the leftmost vertex
of an occurrence of the pattern $1^+2(02)^{+r}1^+$,
then the value of its pendent
neighbours is not modified in step~3.

If $\lambda_{a+1}=1$, then, in step~2, the value of
$\ell_{a+1}^1$ is set to~2 and not modified in step~4.
If $\lambda_{a+1}=2$, then the value of the pendent
neighbours of $v_{a+1}$ is not modified in steps 2 and~4.
Therefore, $f_4^*(v_{a+1})=2$ in both cases.

\item {\it Item~4 of Lemma~\ref{lem:not-greater-than-beta-star}}.\\
Let $v_a\sigma v_b=v_iv_{i+1}\dots v_{i+2r+2}$ be
an occurrence of the pattern $1^+2^-(02^-)^{+r}1^+$,
for some $i$, $0\le i\le k-2r-2$.
In that case, we have $t_\sigma=r$.

If $v_a\sigma v_b$ is
an occurrence of the pattern $1^+2(02)^{+r}1^+$,
the value of the vertices of $\sigma$ are not modified in steps 2 to~4.
Therefore, we have $f_4^*(\sigma)=f_1^*(\sigma)=2(r+1)+r=3r+2=3t_\sigma+2$.

Suppose now that $\sigma$ contains at least one stem having only one pendent neighbour.
In step~3, the value of $\ell_{i+1}^1$ is set to~2 if $\lambda_{i+1}=1$,
the value of $\ell_{i+2r+1}^1$ is set to~2 if $\lambda_{i+2r+1}=1$,
 the value of $\ell_{i+2j+1}^1$, $1\le j\le r-1$, is set to~3 
(and the value of $\ell_{i+2j+1}^2$ is set to~0 if $\lambda_{i+2j+1}=2$), and the value of every trunk is set to~0.
We thus get
$$f_4^*(\sigma)=f_3^*(\sigma)=2+2+3(r-1)=3r+1=3t_\sigma+1.$$

\item {\it Item~5 of Lemma~\ref{lem:not-greater-than-beta-star}}.\\
Let $\sigma=v_iv_{i+1}\dots v_{i+2r+2}$ be
an occurrence of the pattern $02^-(02^-)^{*r}0$,
for some $i$, $1\le i\le k-2r-3$.
In that case, we have $t_\sigma=r+2$.

If $\sigma$ is
an occurrence of the pattern $02(02)^{*r}0$,
the value of the vertices of $\sigma$ are not modified in steps 2 to~4.
Therefore, we have $f_4^*(\sigma)=f_1^*(\sigma)=2(r+1)+r+2=3r+4=3t_\sigma-2$.

Suppose now that $\sigma$ contains at least one stem having only one pendent neighbour.
In step~3, 
 the value of $\ell_{i+2j+1}^1$, $0\le j\le r$, is set to~3 
(and the value of $\ell_{i+2j+1}^2$ is set to~0 if $\lambda_{i+2j+1}=2$), and the value of every trunk is set to~0.
We thus get
$$f_4^*(\sigma)=f_3^*(\sigma)=3(r+1)=3r+3=3t_\sigma-3.$$

\item {\it Item~6 of Lemma~\ref{lem:not-greater-than-beta-star}}.\\
Let $v_0\dots v_{2r+1}$ be an occurrence of the pattern $[2^-(02^-)^{*r}0$
(the case $02^-(02^-)^{*r}]$ is similar, by symmetry).
In that case, we have $t_\sigma=r+1$.

In step~3, 
 the value of $\ell_{2j}^1$, $0\le j\le r$, is set to~3 
(and the value of $\ell_{2j}^2$ is set to~0 if $\lambda_{2j}=2$), and the value of every trunk is set to~0.
We thus get
$$f_4^*(\sigma)=f_3^*(\sigma)=3(r+1)=3r+3=3t_\sigma.$$

\end{enumerate}

By Lemma~\ref{lem:not-greater-than-beta-star},
we know that there exists an optimal non-canonical independent broadcast
$\tilde{f}$  with $\cost(\tilde{f})=\cost(f)$ which satisfies all items
of Lemma~\ref{lem:not-greater-than-beta-star}.
We have proved that the non-canonical independent broadcast $f_4$
constructed in the proof of Lemma~\ref{lem:beta-star} also
satisfies all items
of Lemma~\ref{lem:not-greater-than-beta-star}.
Thanks to Lemma~\ref{lem:all-but-special-case}, we thus have
$$\cost(f)=\cost(\tilde{f})\le\cost(f_4)=\beta^*(CT),$$
which completes the proof.
\end{proof}


We are now able to state our main result, which determines the 
broadcast independent number of any caterpillar with no pair of adjacent trunks.

\begin{theorem}
Let $CT=CT(\lambda_0,\dots,\lambda_k)$ be a caterpillar of length $k\ge 1$, with no pair of adjacent trunks.
The broadcast independence number of $CT$ is then given by:
$$\beta_b(CT)=\max\big\{2(\diam(CT)-1),\beta^*(CT)\big\}.$$
\label{th:main}
\end{theorem}

\begin{proof}
We know by Observation~\ref{obs:2(d-1)} that $\beta_b(CT)\ge 2(\diam(CT)-1)$
and we already observed that the canonical independent broadcast $f_c$ on $CT$
satisfies $\cost(f_c)=2(\diam(CT)-1)$.
According to Lemma~\ref{lem:beta-star}, it is thus enough to prove
that for any optimal non-canonical independent broadcast $f$ on $CT$
with $\cost(f) > 2(\diam(CT)-1)$, $\cost(f)\le\beta^*(CT)$, which directly follows
from Lemma~\ref{lem:beta-star-is-the-best}.
\end{proof}


In several cases, the value of $\beta^*(CT)$ has a simple expression.
Consider for instance a caterpillar $CT$, of length $k\ge 1$, having no trunk.
We then have 
$\beta^*(CT)=\lambda(CT)+n_1(CT)$,
where $n_1$ stands for the number of spine vertices having exactly one pendent
vertex.
Since $\lambda(CT)\ge n_1(CT) + 2(k+1-n_1(CT))=2k+2-n_1(CT)$
(spine vertices have either one or at least two pendent neighbours),
we get $\beta^*(CT)\ge 2k+2$, with equality if and only if $CT$
contains no stem with at least three pendent neighbours.
Since $2(\diam(CT)-1)=2k+2$, we get the following corollary of Theorem~\ref{th:main}.

\begin{corollary}
Let $CT$ be a caterpillar of length $k\ge 1$ having no trunk.
We then have $\beta_b(CT)=2k+2=2(\diam(CT)-1)$ if $CT$ has no stem with 
at least three pendent neighbours,
and $\beta_b(CT)=\lambda(CT)+n_1(CT)$ otherwise. 
\label{cor:no-trunk}
\end{corollary}

Moreover, thanks to Observation~\ref{obs:subgraph}, we 
can also give the broadcast independent number of caterpillars
having adjacent trunks but no stem with 
at least three pendent neighbours.

\begin{corollary}
Let $CT$ be a caterpillar of length $k\ge 1$.
If $CT$ has no stem with at least three pendent neighbours,
then $\beta_b(CT)=2k+2=2(\diam(CT)-1)$. 
\label{cor:many-trunk}
\end{corollary}

Finally, note that if every stem in a caterpillar $CT$ of length $k\ge 1$
with no pair of adjacent trunks has at least three pendent neighbours,
then no pattern involved in the definition of $\beta^*(CT)$ can appear in $CT$.
In that case, since $\tau(CT)\le\left\lfloor \frac{k}{2}\right\rfloor$ and 
$\lambda(CT)\ge 3\left(\left\lceil \frac{k}{2}\right\rceil+1\right)$, we get 
$$\beta^*(CT)=\lambda(CT)+\tau(CT)>2k+2=2(\diam(CT)-1).$$
Therefore, we have:

\begin{corollary}
Let $CT$ be a caterpillar of length $k\ge 1$, with no pair of adjacent trunks.
If all stems in $CT$ have at least three pendent neighbours,
then $\beta_b(CT)=\lambda(CT)+\tau(CT)$. 
\label{cor:0-or-3+}
\end{corollary}

\section{Concluding remarks}
\label{sec:discussion}

In this paper, we studied independent broadcasts of caterpillars and gave an explicit formula for the
broadcast independence number of caterpillars having no pair of adjacent trunks.

This result concerns a quite restricted subclass of the class of trees, but the broadcast independence
number is certainly a difficult parameter  to determine for trees, and probably even for caterpillars.
It should be noticed here that the computational complexity of the decision problem
associated with the broadcast independence number is not known yet, even for trees, although this
question was already posed in~\cite{DEHHH06} and~\cite{H06}.
(The only complexity result about a broadcast parameter, among those introduced in~\cite{DEHHH06}, 
is due to Heggernes and Lokshtanov~\cite{HL06}, who proved that computing the broadcast domination number
$\gamma_b(G)$ of any graph $G$ can be done in polynomial time.)

Finally, we consider that the following questions are of particular interest.

\begin{enumerate}
\item Can we determine the broadcast independence number of caterpillars?
(We should notice here that for caterpillars with adjacent trunks,
Lemmas~\ref{lem:stem-0} and~\ref{lem:end-leaves} still hold, 
while Lemma~\ref{lem:internal-stems} does not.
This explains why we think that this might be a not so easy question.)

\item Can we determine the broadcast independence number of other subclasses of
the class of trees?
In particular, what about $k$-ary trees?

\item Can we characterize the set of caterpillars $CT$ for which
$\beta_b(CT)=2(\diam(CT)-1)$? (Partial answers are given by Corollaries~\ref{cor:no-trunk}
and~\ref{cor:many-trunk}.)

\item More generally, can we characterize the set of trees $T$ for which
$\beta_b(T)=k(\diam(T)-1)$, where $k$ is the maximum size of a set of pairwise antipodal
vertices in $T$?
\end{enumerate}


\end{document}